%% file: blocking_variance.tex
\documentclass[12pt]{article}
\usepackage{amsmath}
\usepackage{amssymb}
\usepackage{amsthm}
\usepackage{graphicx,psfrag,epsf}
\usepackage{psfrag,epsf}
\usepackage{enumitem}
\usepackage{natbib}
\usepackage{url} % not crucial - just used below for the URL 
\usepackage[titletoc,title]{appendix}
\usepackage{array}
\usepackage{bm}
\usepackage{color}
\usepackage{graphicx}
\usepackage{epstopdf}
\usepackage[letterpaper, portrait, margin=1in]{geometry}
\usepackage{lipsum}
\usepackage{listings}
\usepackage{fancyref}
\usepackage{float}
\usepackage[font=footnotesize,labelfont=bf]{caption}
\usepackage{color}
\usepackage{xr}
\newtheorem{theorem}{Theorem}[subsection]

\newtheorem{corollary}{Corollary}[subsection]

\theoremstyle{remark}
\newtheorem*{remark}{Remark}

\newcommand*\taufs{\tau_{\mathcal{S}}}
\newcommand*\taukfs{\tau_{k, \mathcal{S}}}
\newcommand*\tauk{\tau_{k}}
\newcommand*\tausbfs{\tau_{(SMALL),\mathcal{S}}}

\newcommand*\tausbfsj{\tau_{(SMALL), \mathcal{S}, j}}

\newcommand*\tausb{\tau_{(SMALL)}}

\newcommand*\taubigfs{\tau_{(BIG), \mathcal{S}}}

\newcommand*\poobs{Y_{i}^{obs}}
\newcommand*\pot{Y_{i}(t)}
\newcommand*\poc{Y_{i}(c)}
\newcommand*\poz{Y_{i}(z)}

\newcommand*\meanpozk{\bar{Y}_k(z)}

\newcommand*\meanpotkobs{\bar{Y}_k^{obs}(t)}
\newcommand*\meanpockobs{\bar{Y}_k^{obs}(c)}
\newcommand*\meanpozkobs{\bar{Y}_k^{obs}(z)}

\newcommand*\model{\mathcal{F}}

\newcommand*\EE{\mathbb{E}}

\newcommand*\nk{n_k}
\newcommand*\nck{n_{c, k}}
\newcommand*\ntk{n_{t, k}}
\newcommand*\nzk{n_{z, k}}
\newcommand*\numsbunits{n_{sb}}
\newcommand*\numsizej{m_{j}}
\newcommand*\numsb{K_{j}}

\newcommand*\tauestbk{\widehat{\tau}_{(BK)}}

\newcommand*\taukest{\widehat{\tau}_{k}}
\newcommand*\tausbest{\widehat{\tau}_{(SMALL)}}
\newcommand*\taubigest{\widehat{\tau}_{(BIG)}}
\newcommand*\tausbjest{\widehat{\tau}_{(SMALL), j}}

\newcommand*\Stc{S^{2}(tc)}

\newcommand*\Sck{S^{2}_k(c)}
\newcommand*\Stk{S^{2}_k(t)}
\newcommand*\Szk{S^{2}_k(z)}
\newcommand*\Stck{S^{2}_k(tc)}
\newcommand*\sckest{s^{2}_k(c)}
\newcommand*\stkest{s^{2}_k(t)}
\newcommand*\szkest{s^{2}_k(z)}

\newcommand*\sigmatc{\sigma^{2}(tc)}
\newcommand*\sigmack{\sigma^{2}_k(c)}
\newcommand*\sigmatk{\sigma^{2}_k(t)}
\newcommand*\sigmazk{\sigma^{2}_k(z)}

\newcommand*\varestk{\widehat{\sigma}^{2}_{k}}
\newcommand*\varestbkone{\widehat{\sigma}^{2}_{(BK)}}

\newcommand*\varestsbm{\widehat{\sigma}^{2}_{(SMALL/m)}}
\newcommand*\varestsbp{\widehat{\sigma}^{2}_{(SMALL/p)}}
\newcommand*\varestsbs{\widehat{\sigma}^{2}_{(SMALL/s)}}
\newcommand*\varestsbj{\widehat{\sigma}^{2}_{(SMALL), j}}

\usepackage{color}

\begin{document}
\def\spacingset#1{\renewcommand{\baselinestretch}%
{#1}\small\normalsize} \spacingset{1}

%%%%%%%%%%%%%%%%%%%%%%%%%%%%%%%%%%%%%%%%%%%%%%%%%%%%%%%%%%%%%%%%%%%%%%%%%%%%%%

  \title{Insights on Variance Estimation for Blocked and Matched Pairs Designs\thanks{\noindent{\textit{Email: } \texttt{npashley@g.harvard.edu}. The authors would like to thank Guillaume Basse, Avi Feller, Colin Fogarty, Michael Higgins, Luke Keele, and Lo-Hua Yuan for their comments and edits.
We would also like to thank members of Luke Miratrix's and Donald B. Rubin's research labs for their useful feedback on the project and Peter Schochet and Kosuke Imai for insightful discussion of this material.
Finally, we thank anonymous reviewers for their helpful feedback.
The research reported here was partially supported by the Institute of Education Sciences, U.S. Department of Education, through Grant R305D150040.
This material is also based upon work supported by the National Science Foundation Graduate Research Fellowship under Grant No. DGE1745303.
Any opinion, findings, and conclusions or recommendations expressed in this material are those of the authors and do not necessarily reflect the views of the National Science Foundation, the Institute of Education Sciences or the U.S. Department of Education.
}}}
    \author{Nicole E. Pashley\\Department of Statistics, Harvard University  \and Luke W. Miratrix\\Graduate School of Education, Harvard University}
 \maketitle

\bigskip
\begin{abstract}
Evaluating blocked randomized experiments from a potential outcomes perspective has two primary branches of work.
The first focuses on larger blocks, with multiple treatment and control units in each block.
The second focuses on matched pairs, with a single treatment and control unit in each block.
These literatures not only provide different estimators for the standard errors of the estimated average impact, but they are also built on different sets of assumptions.
Neither literature handles cases with blocks of varying size that contain singleton treatment or control units, a case which can occur in a variety of contexts, such as with different forms of matching or post-stratification.
In this paper, we reconcile the literatures by carefully examining the performance of variance estimators under several different frameworks.
We then use these insights to derive novel variance estimators for experiments containing blocks of different sizes.
\end{abstract}

\noindent%
\vfill

\newpage
\spacingset{1.45} 
\section{Introduction}
\label{sec:intro}
Beginning with Neyman and Fisher, there is a long literature of analyzing randomized experiments by focusing on the assignment mechanism rather than some generative model of the data. 
One major family of experimental designs in this literature is the blocked randomized experiment, where units are grouped to hopefully create homogenous collections, and then treatment assignment is randomized within each group \citep[see][]{Fisher1992}. 
Ideally, this process gives a higher precision estimate of the overall average treatment effect, as compared to a completely randomized design.

In the potential outcome causal literature, \citep[as in][]{CausalInferenceText, rosenbaum_2010},\footnote{In particular, we focus on the potential outcomes literature as opposed to the experimental design literature \citep[as in][]{cochran1950, WuHamada}.} much of the prior work on randomized experiments has focused on two forms of blocking: blocking where there are several treated and control units in each block and blocking where there is exactly one treated and one control unit in each block (matched pairs).
See, for example \cite{imai_king_stuart_2008} or \cite{imbens2011experimental} for treatments of large blocks and \cite{abadie2008} or \cite{Imai2008} for treatments of matched pairs.
This literature, for the most part, has a gap: it has not extensively treated the cases where researchers have generated groups of varying size but where there is still only one treated and/or one control in some of the blocks, which we call the ``hybrid design.''
Recent textbooks such as \textit{Field Experiments: Design, Analysis and Interpretation} \citep[][]{gerber_green_text} and \textit{Causal Inference for Statistics, Social, and Biomedical Sciences: An Introduction} \citep[][]{CausalInferenceText} do not propose a clear answer for Neyman-style variance estimation in this case.
While obtaining a point estimate for the overall average treatment effect is straightforward in this context, assessing the uncertainty of such an estimate is not.
Currently one would instead have to turn to Fisher-style permutation tests, which typically rely on constant treatment effect assumptions, or regression-based approaches, which can be biased and usually require assumptions as to the residual error structure.
We build on prior work to fill this gap by providing novel methods for conducting Neyman-style analyses for this more general hybrid design.
The approach to causal inference used in this work also has strong connections to the survey sampling literature, as treated in, e.g., \citet{sarndal2003model} or \citet{CochranWilliamG.WilliamGemmell1977St}.

This gap is important as hybrid experiments with blocks of different sizes, and different numbers of treated and control units within the blocks, can easily arise in many modern social science experiments.
For example, multisite trials in education often have several sites (e.g., districts) with only a few schools in each site.
Many matching methods used in observational studies generate hybrid designs as well. 
For instance, Coarsened Exact Matching (CEM) \citep{iacus2012causal} can lead to many variable-sized blocks, some of which have singleton treatment or control units.
``Full matching,'' which identifies collections of units that are similar on some baseline covariates \citep{hansen_full_match, rosenbaum1991characterization}, creates variable-sized blocks, each with exactly one treated or one control unit.
Our approach allows for a Neyman-style analysis in these contexts.
See Section~\ref{sec:data_exp} for more on these applications.

There are several different models used for Neyman-style causal inference.
The first, the finite sample model, takes the sample of units in the experiment as fixed, using the assignment mechanism as the sole source of randomness.
Other so-called super-population or population models assume that the units in the experimental sample come from some larger population; this can induce additional uncertainty that needs to be accounted for.
With blocking, there is the further complication of how the blocks in the final experimental sample are formed.
There can be fixed blocks in which every unit inherently belongs to one of a finite number of blocks; flexible blocks made by the experimenter once a sample is obtained; and structural blocks that capture natural groupings of units.
There are also several possible sampling mechanisms beyond the classic simple random sampling of units typically presumed, such as sampling from strata corresponding to the blocks or sampling entire blocks rather than individual units.
We believe these variants in how blocks are formed and sampled has caused the gap of the hybrid design: because much of the current literature uses different frameworks tailored to the specific special cases of either large blocks or matched pairs, it is not easily extended as the variance and variance estimators differ across these variants.
As part of our work we carefully outline the common frameworks used and discuss how they are different from each other and how they connect to different types of blocks.
We also analyze the performance of uncertainty estimation for all cases.

Recent work by \cite{fogarty_2017} has also addressed some of these issues.
In particular, Fogarty presents a method for estimating variance with small blocks of variable size, not just matched pairs.
His estimators share some similarities with ours, though they are distinct and we note the difference in bias in Section~\ref{subsec:var_est_perf_fin}.
He also makes explicit the issue of differing results under different population and sampling frameworks by comparing multiple settings.
In our paper, we tackle the issue of creating a cohesive hybrid estimator for experiments with large and small blocks and do not focus on the use of covariates to model treatment effect heterogeneity.

In Section~\ref{sec:setup_notation} we set out our notation and discuss blocked randomization.
We begin with the finite sample framework because it is a building block for the infinite population frameworks. Section~\ref{sec:var_est} provides methods for estimating uncertainty in the case of large blocks, small blocks, and the hybrid of the two, and gives their bias under the finite sample framework.
We then, in Section~\ref{subsec:lit_review}, provide true variance formula and the performance characteristics of the variance estimators for several infinite population frameworks.
Section \ref{sec:sim} contains finite sample simulation studies to illustrate estimator performance and Section \ref{sec:data_exp} illustrates estimation in two data examples.
For clarity in presentation, we have moved the derivations of provided formulae to the Supplementary Material.
To use these methods in practice, we refer readers to our R package, \citep{blkvar}.
Sample scripts demonstrating its use and replicating our simulations are also available.

\section{Overall setup and notation} \label{sec:setup_notation}

We use the Neyman-Rubin model of potential outcomes (Rubin, \citeyear{rubin_1974}; Splawa-Neyman et al., 1923/\citeyear{neyman_1923}).
We assume the Stable Unit Treatment Value Assumption of no differential forms of treatment and no interference between units \citep[][]{rubin_1980}.
Consider an experimental sample of $n$ units. 
In a completely randomized experiment, the entire collection of the units is divided into a treatment group and a control group by taking a simple random sample of $pn$ units as the treatment group and leaving the remainder as control.
In a blocked randomized experiment, our sample is divided into $K$ blocks, formed based on some pretreatment covariate(s), with $n_k$ units in block $k$.
Each block $k$ is then treated as a mini-experiment, with a fixed number of $p_k n_k$ units being randomly assigned to treatment and the rest to control, independently of the other blocks.

The sample average treatment effect (SATE) is the typical estimand in so-called finite sample inference, which takes our sample as fixed, leaving the assignment mechanism as the only source of randomness. 
Under blocking, the SATE within block $k$, for $k=1,...,K$, is
\[\taukfs = \frac{1}{\nk} \sum_{i: b_{i} = k}\big(\pot - \poc\big),\]
where $\pot$ and $\poc$  are the potential outcomes for unit $i$ under treatment and control, respectively, and where $b_i$ indicates the block that unit $i$ belongs to.
The overall SATE (see \cite{CausalInferenceText}, p. 86) is then
\[\taufs = \frac{1}{n} \sum_{i=1}^{n}\big(\pot - \poc\big).\]

In this work, we consider two estimators for the SATE (and later the population average treatment effect), one typically used for complete randomization and one for blocked randomization.
Define the variable $Z_{i}$ as $Z_{i} = t$ if unit $i$ is assigned treatment and $Z_{i} = c$ if unit $i$ is assigned control, for $i=1,...,n$.
Let $\mathbb{I}_{Z_i = t}$ be the indicator that unit $i$ received treatment, $n_{t}$ be the total number of treated units, and $n_{c}$ be the total number of control units.
So, $n_{t} = \sum_{i = 1}^{n}\mathbb{I}_{Z_i = t}$, $n_{c} = n - n_{t}$.
Similarly, let $\ntk$, $\nck$ indicate these values within block $k$.
Define $\poobs = Y_i(Z_i)$ as the outcome we observe for unit $i$ given a specific treatment $Z_i$. 
The blocked randomization estimator is then a weighted average of simple difference estimators for each block
\[\tauestbk = \sum_{k=1}^K\frac{\nk}{n}\taukest, \]
with the
\[\taukest = \frac{1}{\ntk}\sum_{i: b_{i}=k} \mathbb{I}_{Z_i = t} \pot  - \frac{1}{\nck} \sum_{i: b_{i} =k} (1-\mathbb{I}_{Z_i = t}) \poc,\]
$k=1,...,K$, being simple difference estimators within each block.

In general, $\tauestbk$ is unbiased, with
\[ \EE\left[\tauestbk | \mathcal{S}\right] = \taufs, \]
with $\EE\left[ \widehat{M} | \mathcal{S} \right]$ the expected value of some estimator $\widehat{M}$ for a given, fixed, finite sample $\mathcal{S}$ over the blocked randomization.
It is describing and estimating the variance of $\tauestbk$ that is more tricky.
This assessment is the goal of the paper, but first we need to introduce a few more useful concepts.

An important aspect of blocking is how the blocks are formed.
Explicit articulation of block formation will be useful when we discuss asymptotic properties of our estimators and will also be used to differentiate the various population frameworks in Section~\ref{subsec:lit_review}.
We identify three primary ways that blocks are formed: 
\begin{enumerate}[label=(\alph*)]
\itemsep-.2em 
\item \underline{Fixed blocks:} Occurs when the total number of blocks and the covariate distribution of blocks is fixed before looking at the sample. E.g., blocking that occurs on a single categorical covariate.
\item \underline{Flexible blocks:} Occurs when the covariate distribution and total number of blocks may not be known before looking at the sample's covariates.  E.g. if there are many covariates or continuous covariates and matching or discretizing is used to form blocks.
\item \underline{Structural blocks:} Occurs when units have some natural grouping such that the blocks are self-contained. The members of each block are fixed and if a block is represented in the sample, typically all members of that block are in the sample. E.g., twins or classrooms. 
\end{enumerate}
Note that structural blocks are often thought of as clusters.
With clusters, however, treatment assignment is commonly assigned at the cluster level, whereas we are focusing on treatment assigned within cluster.
We use ``structural block'' to clarify this difference.
 
\section{Variance estimation}\label{sec:var_est}

We next discuss how to estimate a blocked estimator's variance, an integral part of obtaining standard errors and confidence intervals.
We discuss from a Neyman-Rubin randomization perspective.
See Supplementary Material~\ref{append:var_est_alt} for a discussion of alternative variance estimators (such as from linear models) that make additional assumptions on the data structure.
We first investigate bias under a finite sample framework and extend to other frameworks in Section~\ref{subsec:lit_review}.

We start by giving the true variance in the finite sample.
To do so, we need some additional notation.
The mean of the potential outcomes for the units in the sample under treatment $z$ for block $k$ is
\[\meanpozk = \frac{1}{n_k}  \sum_{i: b_{i} = k} \poz, \]
the sample variance is
\[\Szk = \frac{1}{n_k-1} \sum_{i: b_{i} = k} (\poz - \meanpozk)^{2}, \]
and the sample variance of the individual level treatment effects is
\[\Stck = \frac{1}{n_k - 1} \sum_{i: b_{i} = k} \Big(\pot - \poc - \taukfs \Big)^{2}.\]

For the finite sample, the variance of $\taukest$ within a block is well known (see \cite{CausalInferenceText, imbens2011experimental}):
\begin{equation}
\text{var}(\taukest|\mathcal{S}) = \frac{\Stk}{\ntk} + \frac{\Sck}{\nck} - \frac{\Stck}{\nk}. \label{eq:var_fs_cr}
\end{equation}
Summing these across the independent blocks, with the weights for block sizes, gives an overall variance of
\begin{align}
\text{var}\left(\tauestbk|\mathcal{S}\right)& = \sum_{k=1}^{K}\frac{\nk^{2}}{n^{2}}\text{var}(\taukest |\mathcal{S})= \sum_{k=1}^{K}\frac{\nk^2}{n^2}\Big(\frac{\Stk}{\ntk} + \frac{\Sck}{\nck} - \frac{\Stck}{\nk}\Big).\label{eq:var_fs_bk}
\end{align}

For blocked experiments, the type of variance estimator one would use in the finite sample depends on the sizes of blocks one has.
In cases where we have at least two treated and two control units in each block, we can directly extend classic results for completely randomized experiments by using them within each block and weighting \citep[see, e.g.,][]{imbens2011experimental, Miratrix2013, mukerjee2018using}.
In particular, we can estimate each variance component of Equation~\ref{eq:var_fs_bk} as
\begin{align}
\varestk  = \widehat{\text{var}}(\taukest) =\frac{\sckest}{\nck} + \frac{\stkest}{\ntk},
\end{align}
with $\szkest$ the sample variance of the units within block $k$ under treatment $z$.
Then we can combine these to get the plug in variance estimator of
\begin{align}
\varestbkone  = \widehat{\text{var}}(\tauestbk) = \sum_{k=1}^{K} \frac{\nk^{2}}{n^{2}}\left(\frac{\sckest}{\nck} + \frac{\stkest}{\ntk}\right).\label{eq:var_est_bk_one}
\end{align}
This gives a conservative estimate due to the dropping of the $\Stck/\nk$ terms.
Some tightening is possible by exploiting features such as differences in the shape of the observed treatment and control outcome distributions; for examples see \citet{aronow2014sharp}, Chapter 6 of \citet{CausalInferenceText}, or \citet{schochet_2016}.
We call this the ``big block'' style of blocking, and the ``big block'' estimator.

For the ``small blocks'' case, where our blocks have only one treated unit or one control unit, we need to use an alternative approach as we cannot estimate the variance for a treatment arm with a single unit.
Our approach is presented below.
To give some background, the analytical problems that arise when estimating the variance in matched pairs experiments, especially when working in the finite sample framework, have been lamented by many statisticians \citep[see, e.g.,][]{imbens2011experimental}. 
The issues arise from the fact that there is no way to estimate the within pair variance with only one unit assigned to treatment and one unit assigned to control in each pair.
Previous work has found conservative estimators, however, which we build on.
For instance, \cite{Imai2008} showed that the standard matched pairs estimator is biased in the finite sample setting and put bounds on the true variance. 
The RCT-Yes R package and documentation \cite[][]{schochet_2016} also provides a conservative variance estimator for the matched pairs design (as well as estimators for blocked designs); this is discussed more in Supplementary Material~\ref{subsubsec:rct_yes_est}.

For a hybrid experiment with both big and small blocks, we combine results to create an overall variance estimator. 

\subsection{Small block experiments with equal size blocks}
When we have small blocks of the same size, we can directly use the usual variance estimator in the matched pairs literature \citep[e.g.,][]{Imai2008} as a variance estimator for $\tauestbk$, no matter what the block sizes are, as also noted by \cite{fogarty_2017}.
This gives a variance estimator of
\begin{align}
\varestsbs  = \frac{1}{K(K-1)}\sum_{k=1}^{K} (\taukest - \tauestbk)^{2}.\label{eq:var_est_sb_same_size}
\end{align}
This estimator directly estimates the variance of the overall block treatment effect estimator, rather than estimating the variance for each individual block and then weighting.
We will see that, depending on the framework used, this estimator can give positively biased estimates if the true $\tau_k$ tends to differ across blocks.

\subsection{Small block experiments with varying size blocks}\label{subsec:var_est_small}

For experiments with small blocks of varying sizes we offer two variance estimators. 
The first directly extends the standard matched pairs estimator by grouping the blocks by size into $J$ groups and using Equation~\ref{eq:var_est_sb_same_size} for each group. 
We then weight and combine to get an overall variance estimator.\\
\noindent \textbf{Stratified Small Block Variance Estimator:}
\begin{align}
\varestsbm = \frac{1}{\left(\sum_{j=1}^J \numsizej \numsb\right)^2}\sum_{j=1}^J (\numsizej\numsb)^2\varestsbj, \label{eq:var_est_sb_m}
\end{align}
where $\numsb$ is the number of blocks of size $\numsizej$ and 
\begin{align}
\varestsbj  = \frac{1}{\numsb(\numsb-1)}\sum_{k: \nk = \numsizej} (\taukest - \tausbjest)^{2}\label{eq:var_est_sb_j}
\end{align}
with $\tausbjest = \sum_{k: \nk = \numsizej}\taukest/ \numsb$.
That is, grouping by the same size allows for using the equal size block estimator above.
While straightforward, this is not ideal because it requires at least two blocks of each size in the overall experiment to estimate each $\varestsbj$.
See Supplementary Material~\ref{append:var_est_sb_m_derv} for further detail.

The second approach allows the variance of all of the small blocks to be estimated at the same time, without requiring multiple blocks of the same size.
\\
\noindent \textbf{Unified Small Block Variance Estimator:}
\begin{align}
\varestsbp = \sum_{k=1}^K \frac{\nk^2}{(n-2\nk)(n+\sum_{i=1}^K\frac{n_i^2}{n-2n_i})}(\taukest - \tauestbk)^{2}. \label{eq:var_est_sb_p}
\end{align}
For $\varestsbp$ to be defined and guaranteed conservative, no one block can make up half or more of the units.
We derived this estimator using the basic form of the matched pairs variance estimator as a weighted sum of the squared differences between the estimated average block treatment effects and the estimated overall average treatment effect.
The weights then come from a simple optimization (see Supplementary Material~\ref{append:var_est_small_p_derv}), and partially account for the different blocks having different levels of precision when estimating the variance of the block-level impacts.
This estimator has similar finite sample properties to the standard estimator for blocks of the same size (Equation~\ref{eq:var_est_sb_same_size}).
In particular, it is also conservative and unbiased when the block average treatment effects are all the same.
When block sizes are all the same, this reduces to the usual matched pairs type estimator.

\subsection{Hybrid experiments}\label{subsec:var_est_hybrid}
When doing variance estimation in a hybrid blocked design, we can split the blocks up into small blocks and big blocks.
Grouping the big and small blocks together allows us to write the causal effect estimand as a combination of two estimands for our two different types of block sizes.
Let there be $\numsbunits$ total units in small blocks in the sample.
Then
\begin{align*}
\taufs & = \frac{n-\numsbunits}{n}\taubigfs + \frac{\numsbunits}{n}\tausbfs
\end{align*}
where 
\[\taubigfs = \frac{1}{n-\numsbunits}\sum_{k: \ntk \geq 2, \nck \geq 2}n_k\tau_k \quad \text{and} \quad \tausbfs = \frac{1}{\numsbunits}\sum_{k: \ntk = 1 \text{ or } \nck = 1}n_k\tau_k.\]

The estimator for the overall treatment effect can also be written as
\begin{align*}
\tauestbk 
& = \frac{n-\numsbunits}{n}\taubigest + \frac{\numsbunits}{n}\tausbest.
\end{align*}

For finite sample inference, we can similarly break down the variance, and estimator of the variance, of $\tauestbk$ because the block estimators are independent due to the block randomized treatment assignment.
\\
\noindent \textbf{Hybrid Variance Estimator:}
\begin{align*}
\widehat{\text{var}}\left(\tauestbk\right)
& = \frac{(n-\numsbunits)^2}{n^2}\widehat{\text{var}}\left(\taubigest\right) + \frac{\numsbunits^2}{n^2}\widehat{\text{var}}\left(\tausbest\right).
\end{align*}
Here we would use $\varestbkone$ (Equation~\ref{eq:var_est_bk_one}) over just the big blocks for $\widehat{\text{var}}\left(\taubigest\right)$ and either $\varestsbm$ (Equation~\ref{eq:var_est_sb_m}) or $\varestsbp$ (Equation~\ref{eq:var_est_sb_p}) over just the small blocks (with the appropriate assumptions for just the small blocks) for $\widehat{\text{var}}\left(\tausbest\right)$.
Thus, when we have small blocks, we can estimate the variance for those small blocks separately and use the usual blocking estimator on the larger blocks, essentially treating these as two separate experiments and combining with appropriate weights in the end.
Alternatively, one could use $\varestsbm$ or $\varestsbp$ for all blocks, but we do not recommend this for the finite sample.

\subsection{Finite sample bias of the variance estimators}\label{subsec:var_est_perf_fin}

In the finite setting all of the above estimators are conservative, and are only unbiased in specific circumstances.
Each block is a miniature complete randomized experiment.
For such experiments, $\varestk$ is known (\citeauthor{CausalInferenceText}, \citeyear{CausalInferenceText},  p. 92; Splawa-Neyman et al., 1923/\citeyear{neyman_1923}) to have bias 
\begin{align*}
\EE\left[\varestk |\mathcal{S}\right] - \text{var}\left(\taukest |\mathcal{S}\right) &= \frac{\Stck}{\nk}.
\end{align*}

If all of the blocks have at least two treated and two control units, we can extend this result to $\varestbkone$ (Equation~\ref{eq:var_est_bk_one}), which has bias 
\begin{align*}
\EE\left[\varestbkone |\mathcal{S}\right] - \text{var}\left(\tauestbk|\mathcal{S}\right) &= \sum_{k=1}^K\frac{\nk}{n^2}\Stck .
\end{align*}
This extends readily to the big block component of the hybrid estimator by only including in the sum those blocks that are big, changing the $n^2$ in the denominator by $(n-\numsbunits)^2$, and weighting appropriately.

For the small blocks of varying sizes, we have two main results.
In presenting these, we assume that the whole sample is made up of small blocks, though, as with the bias of $\varestbkone$, the extension to the small block component of the hybrid estimator is straightforward.
See Supplementary Material~\ref{append:bias_var_est_sb_m_strat_samp} and \ref{append:var_est_small_p_derv} for proofs.  
The first is a Corollary to classic results on matched pairs \citep[see, e.g.,][]{Imai2008}:

\begin{corollary}\label{lemma:var_est_sb_m_fin_samp}
The bias of $\varestsbm$ (Equation~\ref{eq:var_est_sb_m}) under the finite framework is 
\begin{align*}
\EE\left[\varestsbm |\mathcal{S}\right] - \text{var}\left(\tausbest|\mathcal{S}\right) &= \sum_{j= 1}^{J}\frac{\numsb\numsizej^2}{n^2(\numsb-1)}\sum_{k: n_k=m_j}\Big(\taukfs - \tausbfsj\Big)^{2} .
\end{align*}
\end{corollary}
The above extends prior results for $\varestsbs$ for matched pairs (see \cite{Imai2008}, \cite{CausalInferenceText}, p. 227, or, for a more general case, \cite{fogarty_2017}).
$\varestsbm$ is conservative and unbiased when the average treatment effect is the same for all blocks of the same size (similar to the unbiased result from \cite{Imai2008} for $\varestsbs$).

For $\varestsbp$ we have
\begin{theorem}\label{lemma:var_est_sb_p_fin_samp}
The bias of $\varestsbp$ (Equation~\ref{eq:var_est_sb_p}) under the finite framework is 
\begin{align*}
\EE\left[\varestsbp |\mathcal{S}\right]& - \text{var}\left(\tausbest|\mathcal{S}\right)\\
 &=\sum_{k=1}^K \frac{\nk^2}{(n-2\nk)(n+\sum_{i=1}^K\frac{n_i^2}{n-2n_i})}(\taukfs - \tausbfs)^{2},
\end{align*}
assuming no blocks have $n_k \geq n/2$.
\end{theorem}
If the average treatment effect is the same across all small blocks then this estimator is unbiased, and if there is heterogeneity, it is conservative.
This is a distinction from the behavior of the variance estimator suggested in Section 4.2 of \cite{fogarty_2017} for use with variable size small blocks without covariates, in which even with the average treatment effect being the same across all small blocks, the bias is strictly greater than zero.

\begin{remark}
Both small block estimators are conservative, which raises the question of whether one is superior.
The constant in front of each term of the bias of both estimators is of order $\nk^2/n^2$.
Then we expect the bias of $\varestsbm$ to be less than the bias of $\varestsbp$ when the treatment effects of blocks of similar sizes are similar because the variance of impacts within blocks of a given size will be smaller than across all of the blocks.
However, $\varestsbm$ has the drawback that it can only be used when we have at least two blocks of each small size.

The improved potential performance of $\varestsbm$ when there is homogeneity within block sizes does suggest that we could group blocks in some other way if we had prior knowledge of which blocks were most similar.
That is, $\varestsbm$ relies on the blocks being equal size so the weights factor out of the sum to give the expression for the cross-block estimate of variation.
But we could first subdivide our blocks based on some similarity measure and apply $\varestsbp$ to each group, combining the parts with the hybrid weighting approach.
This could make $\varestsbp$ less conservative while maintaining its validity.

\cite{mukerjee2018using} create a general framework for a class of conservative Neyman variance estimators that extends to a variety of causal estimands and estimators in the finite sample context.
Of our estimators, $\varestbkone$ is directly shown as an example in their paper, and $\varestsbm$ can be shown to fall under their framework as well, as we show in Supplementary Material~\ref{supp:muk_et_al_paper}.
The hybrid of these two can then also be included.
Interestingly, it appears that $\varestsbp$ does not fall within their framework, and instead we need to rely on our own methods and derivations.
See Supplementary Material~\ref{supp:muk_et_al_paper} for more details on these connections.

For our estimators, how conservative the estimators are may vary with blocks sizes.
In the case where all blocks are the same size, when we have blocks with $m$ control units and 1 treated unit, as $m$ increases the variance of the treatment effect estimator will decrease, as we are getting a more precise estimate for the control units.
However, the form of the bias of $\varestsbs$ remains the same.
Therefore, with large $m$ the bias of $\varestsbs$ due to treatment effect heterogeneity becomes larger relative to the true variance.
This intuition extends to the variable size case as well.
In these cases alternative variance estimation strategies, such as discussed in Supplementary Material~\ref{append:var_est_alt}, may become more appealing.

The type of blocks also impacts whether the bias of these estimators go to zero as sample size increases.
For instance, one might argue for the use of $\varestsbp$ instead of $\varestbkone$ even if we have big blocks, because the condition for unbiasedness for $\varestsbp$ (that all blocks have the same average treatment effect) could be considered less stringent than for $\varestbkone$ (that there is zero treatment variation within each block).
However, with fixed blocks, the number of units within each block increases as sample size increases and the bias of $\varestbkone$ will go to zero, the standard result, but the bias of $\varestsbp$ will not, unless all of the blocks have the same average treatment effect.
In this case, as the blocks grow to be big, we would use $\varestbkone$.

In the hybrid setting the overall bias will be a weighted sum of the biases for the big and small block components.
Therefore, because the overall weighting depends on the block sizes, having a poor estimator for the small blocks may not have a large effect on the overall bias if small blocks make up only a small proportion of the sample.

There is no way to unbiasedly estimate variance within small blocks without additional structure or covariates. 
If we think that the treatment effects of different strata are not too far apart, then we suggest using one of the previous estimators. 
We at least know that the bias incurred is positive.
However, if we have reason to believe that the treatment effects of different strata will be very far apart, a plug-in estimator, as discussed in Supplementary Material~\ref{append:var_est_alt}, may be more appropriate.

\end{remark}

\section{Infinite Population Frameworks}
\label{subsec:lit_review}

Up to this point we have examined blocking in a finite sample framework, conditioning on the units in the experiment in question.
In the literature, however, blocking has often been examined under a variety of infinite population frameworks.
In particular, the matched pairs literature uses a framework where the blocks themselves are sampled from an infinite population of blocks, whereas the big block literature typically assumes stratified random sampling from a finite number of infinite size strata.
Using different population frameworks will give different answers to important questions of what the true variance of the treatment effect estimate is and what the bias of our variance estimators are.
In this section, we first discuss the literature related to variance estimation for infinite populations, identifying the apparent tensions that exist.
We then systematically discuss different frameworks, deriving the true variance of the treatment effect estimators under each of them.
We also evaluate the bias of the variance estimators introduced in Section~\ref{sec:var_est}.
We focus on infinite superpopulations; finite superpopulations substantially larger than the sample would give similar results.
We explore work pertaining to the use of linear models, such as \cite{cochran_1953} and \cite{winston_lin}, in Supplementary Material~\ref{subsubsec:var_est_lin_reg}.
An important note is that in some cases these sampling schemes are chosen for convenience and that the generalizability of the experiment to the population will depend upon the assumptions made in them being true.
The sampling model may also be considered to serve as a conservative approach to finite sample inference \citep[see][]{ding2017bridging}.

\vspace{0.5cm}
\noindent \textbf{Related work}\\
For matched pairs experiments, \cite{Imai2008} showed that with a superpopulation of an infinite number of structural blocks, specifically matched pairs, from which pairs are randomly sampled, the standard matched pairs variance estimator (Equation~\ref{eq:var_est_sb_same_size}), is unbiased for the population average treatment effect (PATE). 
On the other hand, \cite{imbens2011experimental} showed that the standard matched pairs variance estimator is biased in the setting where we have fixed blocks and units are drawn using stratified random sampling (see Section~\ref{subsec:inf_infstrat_framework} for more on this setting). 
This is a clear example of how the population framework being used matters.
We therefore advise practitioners to carefully consider what population and sampling structure they are assuming and to not simply assume a framework for convenience.

The general blocked design has been previously discussed in various forms. 
 \cite{imbens2011experimental} discussed blocking in the context of a superpopulation with a fixed number of strata from which units are sampled using a stratified sampling method. 
He formed unbiased estimators for the variance in this context, assuming that the blocks each have at least two units assigned to treatment and control. 
These results are similar to finite sample results discussed in Section~\ref{sec:var_est} and will be discussed more in Section~\ref{subsec:inf_infstrat_framework}. 
\cite{imai_king_stuart_2008} analyzed estimation error and variance with the blocked design. 
\cite{scosyrev_2014} also analyzed the blocked experiment in the finite sample and under two sampling frameworks, recognizing that the different settings resulted in different outcomes.
\cite{savje2015performance} analyzed flexible ``threshold'' blocking and made critical points about the importance of block structure and sampling design when analyzing blocked experiments, which we will echo and expand on.

\subsection{Infinite populations in general}\label{subsec:inf_framework}
Inference for the population average treatment effect (PATE) typically takes the sample as a random sample from some larger population, as opposed to inference for the SATE discussed earlier which held the sample of potential outcomes as fixed.
This makes estimation an implicit two-step process, estimating the treatment effect for the sample and extrapolating this estimate to the population.
Frequently, in fact, the estimators themselves are the same as for finite sample inference even though the estimands are different.

Define the PATE as
\[\tau = \EE[\pot - \poc|\model],\]
where $\model$ both indicates the block type and sampling framework.
This is the same as the direct average of the unit-level treatment effects for all of the units in the population, as is commonly used \citep[see][p. 99]{CausalInferenceText}, as long as our sampling mechanism is not biased.
Here we will only consider frameworks where the sampling scheme provides a sample that, on average, has the same average treatment effect as the population but note that bias from the sampling mechanism can be fixed using weighting if the sampling mechanism is known \citep[see][]{miratrix_survey}.

Under blocking, the PATE within block $k$ is
\begin{align*}
\tauk &=\EE[\pot - \poc|b_{i} = k, \model],
\end{align*}
where, again, $b_i$ indicates the block that unit $i$ belongs to. It is possible that $k$ indexes a (countably) infinite set of blocks in the case of some infinite population models.

Overall, using the law of total expectation and variance decompositions, we can generally obtain the properties of our estimators with respect to population estimands by first obtaining expressions for a finite sample and then averaging these expressions across the sampling distributions.
In other words, we heavily exploit $\EE\left[ \widehat{M} |\model \right] = \EE\left[\EE\left[ \widehat{M} | \mathcal{S} \right] | \model \right]$, where $\mathcal{S}$ is a sample obtained from $\model$, our population and sampling framework. 
Under any unbiased framework $\model$, we have the typical result \citep[e.g. see][]{imbens2011experimental}
\[  \EE\left[\tauestbk | \model\right] = \EE\left[\taufs | \model\right] = \tau .\]

There are several different frameworks that one might assume.
These can generally be characterized by two primary features: the block types, which also dictates the population strata structure, and the sampling scheme.
Note that the term strata is used for the population here analogously to blocks in the sample.
We may obtain a sample using simple random sampling and then form blocks based on covariates post-sampling and pre-randomization, i.e. flexible blocks.
Or we may have fixed blocks (e.g. blood types) and use stratified sampling where we sample units from each population stratum.
Finally, we may have structural blocks and conceptualize a population of an infinite number of these blocks (e.g. schools in an ``infinite'' population of schools) from which we randomly select a fixed number of blocks.
As we show next, the bias of the variance estimators can differ depending on the framework assumed.
We refer to frameworks using their sampling method as a shorthand, leaving the block type and population structure implicit.

\subsection{Simple random sampling, flexible blocks}\label{subsec:inf_srs_framework}

In this framework, denoted $SRS$, units are sampled at random, without regard to block membership, from the population.
In this context, we focus on the use of flexible blocks, e.g. blocking using clustering on a continuous covariate or based on observed covariates in the sample obtained.
Structural blocks do not make sense in this framework (e.g. one would always sample pairs of twins not individuals who are twins if we wish to run a twin study) and fixed blocks give rise to difficulties when the sample does not have units from all population strata.
For blocked experiments with fixed blocks in this framework, see \cite{scosyrev_2014}.

The variance in this framework, using the basic variance decomposition, is
\[\text{var}(\tauestbk|SRS) = \EE\left[\sum_{k=1}^K\frac{\nk^2}{n^2}\left(\frac{\Sck}{\nck}+\frac{\Stk}{\ntk}-\frac{\Stck}{\nk}\right)\Big|SRS\right]+\text{var}\left(\taufs|SRS\right).\]
The expectation is across the sampling and blocking process.
SRS denotes the simple random sample and subsequent blocking of sampled units.

In this context we have an unbiased variance estimator if we have all big blocks:
\begin{theorem}
The variance estimator
 \begin{align}
 \widehat{\sigma}^2_{SRS} = \sum_{k=1}^{K}\frac{\nk(\nk-1)}{n(n-1)}\left(\frac{\sckest}{\nck}+\frac{\stkest}{\ntk}\right) + \sum_{k=1}^K \frac{\nk}{n(n-1)} \left(\taukest - \tauestbk\right)^2 \label{eq:var_est_srs_bk}
 \end{align}
 is an unbiased estimator for $\text{var}(\tauestbk|SRS)$, if $\nck \geq 2$ and $\ntk \geq 2$.
 \end{theorem}
 
See Supplementary Materials~\ref{append_var_srs} for a derivation.
The first term in the estimator looks similar to our usual big block estimator and captures part of the first term in our variance decomposition.
The second term looks similar to our proposed small block estimator and accounts for the rest of the variation.
While very similar to the estimator found in \cite{scosyrev_2014}, we have made adjustments to achieve unbiasedness of the estimator  whereas \cite{scosyrev_2014} focuses on consistency.
\cite{scosyrev_2014} also works with fixed blocks where the number of blocks is assumed known before sampling and weights are used to match the sample to the population proportions, as opposed to flexible blocks which allow random numbers of blocks that are created post-sampling.

\begin{remark}
If we na\"ively use $\varestbkone$ (Equation~\ref{eq:var_est_bk_one}) our bias will be
\[\EE\left[\varestbkone|SRS\right] - \text{var}(\tauestbk|SRS) = \frac{1}{n} \EE\left[\sum_{k=1}^K\frac{\nk}{n}\Stck-\Stc \Big|SRS\right],\]
where $S^2(tc)$ is the sample variance of individual level treatment effects across the whole sample.
This result follows from the derivations in Supplementary Materials~\ref{append_var_srs} and it implies that $\varestbkone$ could be anti-conservative in this setting if there is generally treatment variation across samples (making $\Stc > 0$), but units put within the same block are nearly identical in terms of impacts (making $\Stck \approx 0$).
This could happen when the experimenter is successfully making homogenous blocks.

Similarly, if we use either of the small block variance estimators, the bias will be the difference between the expected finite sample bias for those estimators (which depends on treatment effect heterogeneity between blocks) and $\EE\left[\Stc\Big|SRS\right]/n$, which corresponds to treatment effect heterogeneity across the whole population.
Therefore whether these estimators are conservative or not depends upon the structure of the population and how the blocks are formed.
\end{remark}

\subsection{Stratified sampling, fixed blocks}\label{subsec:inf_infstrat_framework}

In the ``stratified sampling'' framework, denoted $\model_1$, there are K fixed strata of infinite size in the population.
Then $\nk$ units are randomly sampled from strata $k$ (i.e., stratified random sampling is used).
Here we have fixed blocks.
We assume that $\nk$ is fixed and that $\nk/n$ is the population proportion of units in stratum $k$, for simplicity.
Otherwise, a  weighting scheme, as mentioned in Section \ref{subsec:inf_framework}, would be needed to create an unbiased estimator of the direct average of treatment effects in the population.
This is the framework used in  \cite{imbens2011experimental} and \cite{Miratrix2013}, who show the following result under equal proportions treated within each block, which simplifies the weights.

As in the finite sample, overall variance is a weighted sum of within block variances:
\begin{align}
\text{var}(\tauestbk|\model_1) &  = \sum_{k=1}^{K}\frac{\nk^{2}}{n^{2}}\text{var}\Big(\hat{\tau}_{k}|\model_{1}\Big) = \sum_{k=1}^{K}\frac{\nk^{2}}{n^{2}}\Big(\frac{\sigmack}{\nck} + \frac{\sigmatk}{\ntk}\Big) \label{eq:var_strat_samp_bk},
\end{align}
with $\sigmazk$ the population variance of the potential outcomes under treatment $z$ in strata $k$.

As noted in \cite{imbens2011experimental}, the variance estimator of big blocks, $\varestbkone$ (Equation~\ref{eq:var_est_bk_one}), is unbiased in this framework.
The estimators for the variance of the small blocks, however, can have bias.
We have two results pertaining to this.
For presentation of results for small blocks, we assume that all blocks in the sample are small but the results extend directly to just the small block component of the hybrid variance estimators.

First, as with the finite sample, we can extend results for $\varestsbs$ \citep[see][]{imbens2011experimental} to $\varestsbm$.

\begin{corollary}\label{lemma:bias_var_est_sb_m_strat_samp}
The bias of $\varestsbm$ (Equation~\ref{eq:var_est_sb_m}) under the stratified sampling framework is 
\begin{align*}
\EE\Big[\varestsbm|\model_1\Big] - \text{var}(\tausbest|\model_1) &= \sum_{j= 1}^{J}\frac{\numsb\numsizej^2}{n^2(\numsb-1)}\sum_{k: n_k=m_j}\Big(\tauk - \tausb\Big)^{2} .\end{align*}
\end{corollary}

As with finite sample inference, this shows that $\varestsbm$ is a conservative estimator unless the average treatment effect is the same across all small blocks of the same size, in which case it is unbiased.
See Supplementary Material~\ref{append:bias_var_est_sb_m_strat_samp} for the derivation.

Second, for our new variance estimator we have the following result:
\begin{corollary}\label{lemma:bias_var_est_sb_p_strat_samp}
The bias of $\varestsbp$ (Equation~\ref{eq:var_est_sb_p}) under the stratified sampling framework is 
\begin{align*}
\EE&\Big[\varestsbp|\model_1\Big] - \text{var}(\tausbest|\model_1)\\
 &= \sum_{k=1}^K \frac{\nk^2}{(n-2\nk)(n+\sum_{i=1}^K\frac{n_i^2}{n-2n_i})}\Big(\tauk - \tausb\Big)^{2},
\end{align*}
assuming no blocks have $n_k \geq n/2$.
\end{corollary}
This shows that $\varestsbp$ is also a conservative estimator (given no block makes up more than half the sample) and it is unbiased when the average treatment effect is the same across all small blocks.
See Supplementary Material~\ref{append:var_est_small_p_derv} for a derivation.

\subsection{Random sampling of strata, structural blocks}\label{subsec:inf_finstrat_framework}

In the ``random sampling of strata'' framework, denoted $\model_2$, there are an infinite number of strata of finite size, i.e. an infinite number of structural blocks.
$K$ strata are then randomly chosen to be in the sample and randomization is done within each of the sample blocks.
This setting, with equal block sizes, is often used in the matched pairs literature, such as in \cite{Imai2008}.

Within this framework, which blocks are included in the sample is itself random.
Therefore, the variance estimator needs to capture not only the within strata variance but also the variance due to which strata are chosen to be in the sample.
Furthermore, if the block sizes vary, the total number of units is random which introduces additional complexities.

For the more general variable-size version of this framework, the variance of $\tauestbk$ is
\begin{align}
\text{var}\left(\tauestbk|\model_{2}\right) 
&= \EE\left[\sum_{k:B_k=1}\frac{n_k^2}{n^2}\left(\frac{\Sck}{\nck} + \frac{\Stk}{\ntk} - \frac{\Stck}{\nk}\right)|\model_{2}\right] + \text{var}\left(\tau_{\mathcal{S}}|\model_{2}\right), \label{eq:var_inf_strat_bk}
\end{align}
where $B_k$ is the indicator that stratum $k$ is included in the sample, with $B_k = 1$ indicating sample membership and $B_k = 0$ otherwise.

When blocks are of the same size, we can simplify the expression with $\frac{\nk^2}{n^2} = \frac{1}{K^2}$, which is no longer random.
If we have all blocks of the same size, then we can rewrite $\varestsbs$ (Equation~\ref{eq:var_est_sb_same_size}) using sample inclusion indicators as
\[ \varestsbs  = \frac{1}{K(K-1)}\sum_{k} B_k(\taukest - \tauestbk)^{2}, \]
and this is an unbiased estimator for $\text{var}(\tauestbk | \model_2)$.
This is simply the variance of the estimated block effect in the sample.
\cite{Imai2008} showed that this estimator is unbiased in this setting with an infinite population of matched pairs.
See Supplementary Material~\ref{append:bias_var_est_sb_m_inf_strat} for the proof of this result extended to other small block types of equal size.

Variance estimators when the strata vary in size are more complicated.
In particular, under this framework there is a chance that there is only a single block of a given size, making the first variance estimator infeasible.
If we condition on the number of strata drawn of each possible strata size, assuming that there are multiple strata of the each size in the sample, we obtain the following Corollary:
\begin{corollary}\label{lem:rand_samp_strat_varsbm}
In the conditioned case, assuming it is defined, $\varestsbm$ (Equation~\ref{eq:var_est_sb_m}) is an unbiased estimator for $\text{var}(\tauestbk | \model_2)$.
\end{corollary}
This result can be seen directly from the results in Supplementary Material~\ref{append:bias_var_est_sb_m_inf_strat}.

Alternatively, if we are willing to assume that block size is independent of treatment effect, then we have the following more general result:
\begin{theorem}[Unbiasedness of $\varestsbp$ given independence]\label{theorem:var_sbp_rand_samp}
In the random sampling of strata setting where block sizes are independent of block average treatment effects, $\varestsbp$ (Equation~\ref{eq:var_est_sb_p}) is an unbiased estimator for $\text{var}(\tauestbk | \model_2)$, assuming no blocks have $n_k \geq n/2$.
\end{theorem}
The proof is in Supplementary Materials~\ref{append:varsbp_rand_samp}.

\begin{remark}
We may also consider an infinite number of strata of infinite size, as is commonly used in multisite randomized trials.
This is the setting considered in \cite{schochet_2016} and the RCT-YES software \citep[][]{schochet_2016} estimator discussed in Supplementary Material~\ref{subsubsec:rct_yes_est} could be used.
The sampling scheme then has two steps: first sample the strata, then sample units from the strata.
To discuss variance, we need to add a bit of notation.
Let $\tau_{\mathcal{S}}^*$ denote the expectation of the treatment effect estimator given the blocks in the sample.
That is, we fix which strata are in the sample and take the expectation over the sampling of units from the infinite size strata.
So conditioning on which strata are in the sample we are in a stratified sampling set up.
Let this framework be denoted by $\model_{3}$.
Then the variance of $\tauestbk$ is
\begin{align*}
\text{var}\left(\tauestbk|\model_{3}\right) 
&= \EE\left[\sum_{k:B_k=1}\frac{n_k^2}{n^2}\left(\frac{\sigmack}{\nck} + \frac{\sigmatk}{\ntk}\right)|\model_{3}\right] + \text{var}\left(\tau_{\mathcal{S}}^*|\model_{3}\right).
\end{align*}
It is straightforward to extend the results of Corollary~\ref{lem:rand_samp_strat_varsbm} and Theorem~\ref{theorem:var_sbp_rand_samp} to this case.
\end{remark}

\subsection{Discussion}\label{subsec:tru_var_disc}

While the variance formulas that we presented above share a similar structure with each other and the finite sample forms, there are important differences.
In the finite sample framework (Equation~\ref{eq:var_fs_bk}), there is a term regarding treatment effect variation that reduces the variance due to the correlation of potential outcomes.
This term is retained in the random sampling of strata framework of Section \ref{subsec:inf_finstrat_framework} but not in the stratified sampling framework of Section \ref{subsec:inf_infstrat_framework}.
This difference in the true variance implies that different variance estimators may be more appropriate in different settings.
It also suggests comparisons of blocking to complete randomization under these different assumptions will also diverge; for further discussion on this, see \cite{pashley2020bkcr}.
In fact, this difference explains much of apparent discrepancy between the matched pairs literature and the blocking literature.

Relatedly, different variance estimators can have different amounts of bias depending on the framework being used. 
The small blocks estimators ($\varestsbm$ and $\varestsbp$) in the finite sample and the stratified sampling framework are unbiased if the average treatment effect is the same across all of the small blocks (or all of the small blocks of the same size for $\varestsbm$) and otherwise are more conservative as the variance of the average treatment effects across blocks increases.
For the infinite number of strata framework, under some assumptions all of our small block variance estimators are unbiased. 
We have no small block estimator that is guaranteed to be unbiased or conservative for the simple random sampling (flexible block) framework, though we present one for big blocks.

The big blocks estimator ($\varestbkone$) in the finite sample is unbiased if the treatment effect is additive within each block and otherwise depends on the treatment effect heterogeneity within each block.
In the stratified sampling framework, however, $\varestbkone$ will be unbiased.

Overall, only the framework of Section~\ref{subsec:inf_finstrat_framework} of sampling structural blocks, with the additional assumption of independence of impact and block size given there, has unbiased variance estimators for a mixture of big and small blocks.
This means that, without additional assumptions allowing for plug-in approaches, the hybrid estimators, where possible, will always be conservative.

\section{Simulations}\label{sec:sim}
We compare different estimators of the variance for hybrid blocked experiments where there are a few big blocks and many small blocks in a finite sample context.
We explore a context where 50\% of our units are in small blocks, each with only one treated unit, and the remainder are in big blocks with at least two treated units.
None of our blocks have many treated units due to only having approximately 20\% of the units treated overall.
(The 20\% was approximate in order to create varying size small blocks to see the different performance of the hybrid estimators.)
We have 15 blocks with sizes ranging from 3 to 20.

The simulations presented here are for the finite sample framework, as it is both a common mode of inference as well as a core building block to the population frameworks.
These results, however, are largely applicable to these other settings.
For instance, the biases for the small blocks variance estimators have the same form for the finite sample and the stratified sampling frameworks.

We considered our two hybrid estimators, which correspond to estimating the variance of the small blocks two different ways.
We also considered two regression estimators: the HC1 sandwich estimate \citep{hinkley1977jackknifing} from a linear model with fixed effects and no interaction between treatment indicator and blocking factor, and the standard variance estimate (inverse Fisher information) from a weighted regression, weighting each unit by the inverse probability of being assigned to its given treatment status in its block, multiplied by the overall proportion of units in its treatment group (this is a variant of the approach in \citet{gerber_green_text}; see also \citet{miratrix_weiss}).\footnote{There are actually different weighting approaches one can use in regression adjustment; in particular one can use precision weighting or survey weighting. In additional explorations we examined survey weighting as implemented by \texttt{svyglm}, and found these other options generally performed more poorly, with some approaches resulting in substantial underestimation of variance and others having a great deal of inflation.}
Note that the HC1 estimator is the ``robust'' estimator used in Stata \citep{statacorp2017stata} for estimating standard deviations.

In our simulations, we varied both to what extent blocking successfully separated units based on their potential outcomes under control and also on their treatment effects.
The average potential outcome under control and the average treatment effect for each block were both negatively correlated with block size, so that smaller blocks had larger control potential outcomes and larger treatment effects.
The correlation of potential outcomes within blocks was also varied between $\rho = 0, 0.5, \text{ and }1$.
See Supplementary Material~\ref{append:more_sim} for more on the data generating process.

We compared all of the variance estimators to the actual variance of the corresponding blocking treatment estimator in Figure \ref{plot:rel_bias} by looking at the percent relative bias ($[\text{mean}(\hat{\sigma}^2_*) -\text{var}(\tauestbk|\mathcal{S})]/\text{var}(\tauestbk|\mathcal{S})$).\footnote{We compare all estimators to the variance of $\tauestbk$ to put everything on the same scale, even though the sandwich estimate for a linear model is estimating the variance of the linear model estimator, which is not generally the same.}
The variation due to changing the between block difference in the mean of control potential outcomes was found to be minimal so we average these differences on the plots.
The two hybrid estimators, the one using $\varestsbm$ (Equation~\ref{eq:var_est_sb_m}) ($\text{Hybrid}_{\text{m}}$) for the small blocks and the one using $\varestsbp$ (Equation~\ref{eq:var_est_sb_p}) ($\text{Hybrid}_{\text{p}}$) for the small blocks, outperform the linear model estimators, especially as the treatment effect variation across the blocks increases.
We see that $\text{Hybrid}_{\text{m}}$ also has lower bias than $\text{Hybrid}_{\text{p}}$ as treatment heterogeneity increases.
This is because the value of treatment effects are correlated with block size and $\varestsbm$ groups variance estimation by block size.
Weighted regression performance was generally similar to that of $\text{Hybrid}_{\text{p}}$, although slightly anti-conservative for samples with low treatment heterogeneity when $\rho = 1$.

For discussion of the variance of the variance estimators, see Supplementary Material~\ref{append:var_est_var}.
The variance estimators' variances were found to be comparable, with weighted regression generally the most stable.

When comparing the performance of estimators, there is an important note about the linear model estimator: the sandwich estimate for a linear model is associated with a different treatment effect estimator than the others.
In particular, a linear model with fixed effects is estimating a precision weighted estimate of the treatment effect across the blocks. 
It is well known that as treatment heterogeneity increases, this estimator can become increasingly biased.
See, \cite{raudenbush2020randomized} for a longer discussion on this and related estimators.
This is not an issue for the weighted regression which, similar to adding interactions between treatment and block dummy variables, will recover $\tauestbk$.

\begin{figure}[h!]
\centering
\includegraphics[scale=1.1]{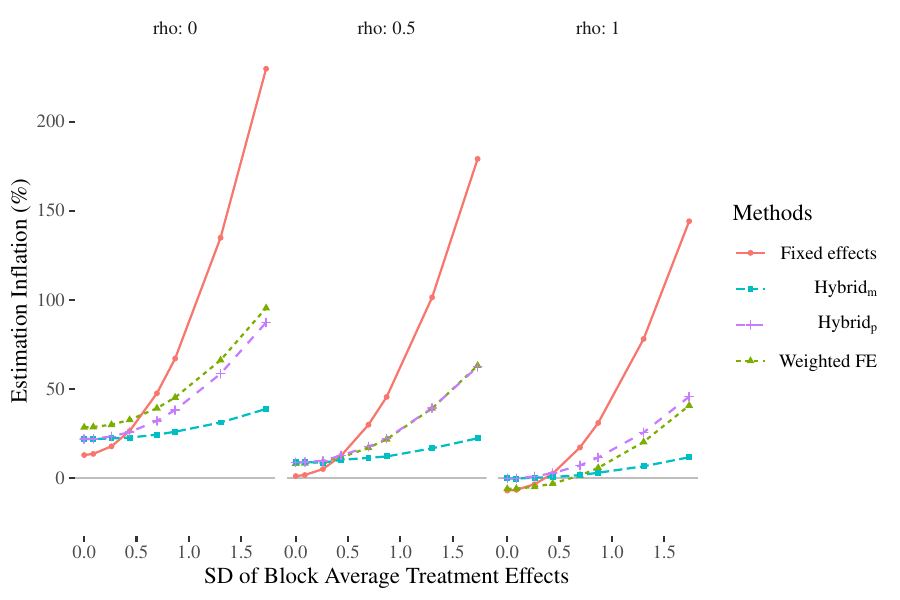}
\caption{Simulations to assess variance estimators' relative bias as a function of treatment variation across blocks. 
Each column represents a different value of $\rho$, with values denoted at the top of the graph.
The x-axis shows the standard deviation of block treatment effects.
Dots indicate average over changes in control means for specific finite samples.
FE stands for fixed effects.
}
\label{plot:rel_bias}
\end{figure}

\section{Data Example}\label{sec:data_exp}
One area where analysts are often faced with many small blocks of varying sizes is found in the matching literature. 
In particular, full matching  (see \cite{hansen_full_match}, \cite{rosenbaum1991characterization}) finds sets of similar units, with either one treated and several control or vice versa, that could be considered as-if randomized. 
After matching, a researcher could then analyze these data using permutation tests and associated sensitivity checks (see, e.g. \cite{rosenbaum_2010}), but in this context generating confidence intervals or standard errors using permutation inference would typically rely on a constant treatment effect assumption across the blocks. 
One might alternatively wish for a Neyman-style randomization analysis such as would be typically done for large block experiments to obtain inference for the average effect in the presence of treatment variation.
The average treatment effect estimate is easy to obtain; it is the uncertainty estimation that causes the trouble.
Our small block variance estimators fills this gap.
To illustrate, we analyze a data set from the National Health and Nutrition Examination Survey (NHANES) 2013-\citeyear{nhanes} given in the \texttt{CrossScreening} package \citep[][]{crossscreening} in \texttt{R} statistical software \citep{Rsoft}.
This data set was also used by \cite{zhao2017cross} to analyze the effect of high fish consumption (defined as 12 or more servings of fish or shellfish in the previous month) versus low fish consumption (defined as 0 or 1 servings of fish or shellfish in the previous month) on a number of biomarkers.

Although \cite{zhao2017cross} analyzed numerous outcomes, we focus on a measure of mercury (LBXGM), converted to the $\text{log}_2$ scale, as a simple illustration of our methods.
We use unrestricted full matching to obtain a set of all small blocks of varying size.
As in \cite{zhao2017cross}, we matched on smoking, age, gender, race, income, and education.
We used Bayesian logistic regression through the \texttt{brglm} package \citep[][]{brglm} and \texttt{optmatch} \citep[][]{optmatch} in \texttt{R} \citep{Rsoft}.
This resulted in 197 blocks with only one treated or one control unit in each.
Sizes of blocks ranged from 2 to 47.
This type of matching would fall into the category of flexible blocks, and we here focus on estimation of the SATE for the finite sample.

\begin{table}
\centering
\begin{tabular}{lcc|cc}
 & \multicolumn{2}{c|}{NHANES} & \multicolumn{2}{c}{Lalonde}  \\
Estimator & Estimate & $\widehat{SE}$ &  Estimate & $\widehat{SE}$ \\
\hline
Hybrid blocking with $\varestsbm$ 			& 2.45 &	N/A 		& \$560 & \$570 \\
Hybrid blocking with $\varestsbp$ 			& 2.45 &	0.20 	& \$560 & \$606 \\
Weighted regression				 			& 2.45 & 0.11	& \$560	& \$560 \\
Fixed effects regression (HC1) 				& 2.75 & 0.13 	& \$425	& \$601 \\
\end{tabular}
\caption{Results of NHANES (full matching), and Lalonde (CEM) for different estimation strategies.}
\label{table:data_ex}
\end{table}

There were some block sizes that were unique, so the hybrid estimator with $\varestsbm$ could not be used.
Alternate forms of full matching could potentially avoid this concern: full matching can include additional restrictions, such as using only a portion of the control group or exact matching on some important covariates, which could make the block sizes more homogenous \citep{optmatch}; for simplicity we do not explore this here.
The blocking treatment effect estimate ($\tauestbk$) was 2.45 but using a fixed effects model with no interaction the treatment effect estimate was 2.75.
Looking at Table~\ref{table:data_ex}, we see that our hybrid estimator using $\varestsbp$ gave a much larger variance estimate (relative to the scale of the precision estimates) than the two linear model based variance estimators.

A second method for analyzing observational datasets where our variance estimators could be useful is coarsened exact matching (CEM).
CEM coarsens covariates used to match and then exactly matches to these coarsened variables \citep[][]{iacus2012causal}.
We follow the example from the vignette of the \texttt{cem} package \citep[][]{cem} in \texttt{R} using the most automated version of CEM on the classic LaLonde data set \citep[][]{lalonde1986evaluating}, available in the \texttt{cem} package.
This data set consists of individuals who received or did not receive a job training program with the outcome of interest as earnings in 1978.
We use the unmodified version of the LaLonde dataset, but otherwise follow the automated process for CEM laid out in the vignette to create blocks (we do not follow the analysis).
This resulted in the creation of 69 blocks, some small and some big, with some ungrouped units being dropped.
The blocking treatment effect estimate ($\tauestbk$) was  \$560 but using a fixed effects model with no interaction the treatment effect estimate was \$425.
From Table~\ref{table:data_ex}, the precision estimates from all methods were similar though, again, the hybrid estimator using $\varestsbp$ was the largest and likely the most conservative.

\section{Discussion}\label{sec:disc}

Blocking can be viewed under a wide variety of population frameworks ranging from a fixed, finite-sample model to one where we envision the units as being sampled from a larger population in pre-set groups.
Because different types of blocking tend to use different frameworks, there has not been good guidance on how to proceed when faced with some singleton units in some blocks and not in others.

We have worked to bring the different frameworks together in order to compare them systematically.
We identified and compared the true variance of a blocking-based estimator under multiple settings, and created corresponding estimators of the impact estimator's variance.
We also provide simple, model-free variance estimators for two types of experiments that have not received much attention: blocked experiments with variable-sized blocks containing singleton treatment or control units, and hybrid blocked experiments with large and small blocks combined. 
These contexts are quite common, frequently appearing in, for example, the matching literature.
We analyzed the performance of both our new variance estimators and the classic variance estimators under different frameworks, identifying when they are unbiased or conservative.
This investigation again illustrates how different sampling frameworks and block types can impact assessments of an estimator's performance.

Future work includes extending these results to other population settings and sampling methods, in particular finding small block estimators for the setting of constructing blocks post-sampling and pre-randomization. 
Variance estimation is also a missing and needed piece in post-stratification research, as noted in \cite{Miratrix2013}.
Although conditional answers for post-stratification would correspond to the estimators presented in this work, the unconditional case remains an open extension.

\bibliographystyle{apalike}
\bibliography{poststratref}{}

%\begin{appendices}
\appendix
\include{blocking_variance_appendices}

%\end{appendices}

\end{document}

%% file: blocking_variance_appendices.tex
\begin{center}
{\bf \Large  Supplementary Material\\for\\``Insights on Variance Estimation for Blocked and Matched Pairs Designs''}

\vspace{0.5cm}
{ \large Nicole E. Pashley and Luke W. Miratrix}
\end{center}

\medskip

These Supplementary Material primarily contain detailed derivations of the results in the main paper, as well as some additional results and discussion. 
We first give some detail on the notational elements in the main paper. We then proceed with a non-technical discussion of alternate estimators for variance estimation.
We next provide additional results and further details on simulations.
We finally give derivations for the results in the main paper.
More detailed proofs for some of these sections are available upon request.

\addcontentsline{toc}{section}{Introduction}

\section{Alternative strategies for variance estimation}\label{append:var_est_alt}
In the main paper we examined strategies for variance estimation that put no structure on how the individual blocks may differ from each other.
At root, the focus is on estimating the residual variance of units around their block means, and aggregating appropriately.

This section discusses alternatives along with when they might be more or less appropriate.
The first two subsections describe estimators that require model assumptions.
These estimators may perform well in certain circumstances (i.e., those where the model assumptions hold), but rely on assumptions that we do not make in our analysis.
They can therefore perform very poorly under misspecification.
The third subsection describes an estimator used for matched pairs that is proposed in the RCT-YES software documentation.
This RCT-YES estimator assumes a specific population model that we do not consider in our paper.

\subsection{Linear regression}\label{subsubsec:var_est_lin_reg}
Perhaps the most common method of estimation for randomized trials is to simply fit a linear model to the data with a treatment indicator and a dummy variable for each block.
If there is no interaction between the treatment and block dummies, this approach will produce a precision-weighted estimate of the treatment effect, with an overall implicit estimand of a weighted average of the average impacts within each block, weighted by the estimated block precision under a homoscedasticity assumption.
If there is impact variation correlated with this precision, then this precision-weighted estimand could be different than the overall ATE, resulting in a biased estimator.
Furthermore, if blocks have different proportions of treated units and different sizes, this weighting might not correspond to any easily interpretable quantity.
As pointed out by \citeauthor{freedman_2008b}, this regression model is also not justified by randomization, which results in complications with the corresponding standard errors (\citeauthor{freedman_2008a},~2008a,b).
Unfortunately, however, this approach is likely the most common in the field.

The above issues are, in part, repairable. 
\cite{winston_lin} shows that the estimator from a linear model including interactions between the treatment indicator and block dummies is unbiased.
In fact, this estimator is equivalent to the blocked estimator presented in this paper.
The question is then how to estimate the standard errors from within the ordinary least squares framework.
Lin advocates a Huber-White sandwich estimator for the general covariate case, but these have problems when the blocks have single treated or single control units.
In particular, several variants of these estimators, such as HC2 and HC3, will not even be defined due to the characteristics of the corresponding design matrix.
The HC0 estimator can still be heavily biased if there is systematic heteroscedasticity across the blocks.
\cite{gerber_green_text} (p. 116-117) advocate a weighted estimator, but this can also fail in the presence of blocks with singleton treated or control units.

\subsection{Pooling variance estimates}\label{subsec:var_est_plg}
As we have seen, the driving idea behind the big block estimator to get an overall variance is to obtain variance estimates for all of the block specific estimates, and then combine them in a weighted average.
This does not work for small blocks, as we do not have enough units to estimate variances of one of the treatment arms.
If we had such estimates, however, we could aggregate as before.

One way forward is to use the variance estimates in other blocks to estimate the variances in the intractable blocks.
For example, if treatment effects within block were considered constant, one could use the variance estimate of the larger of the treatment or control group of the block as the variance of the other treatment arm.
Alternatively, if given a means of assessing how similar blocks are in terms of their variance, one could simply use the variance of the closest big block for each small block.
This typically requires some assumptions that, based on covariate values of the blocks, the variances of the potential outcomes are the same or similar.
Similarly, Abadie and Imbens created an estimator of the variance for matched pairs that involves pairing the closest matched pairs and creating a pooled variance estimator for the two blocks together \citep[][]{abadie2008}.
They found that their estimator was asymptotically unbiased given certain conditions, such as the closeness of pairs increasing as the sample size grew.
Although the asymptotic results derived in their paper are not necessarily appropriate here, this could be a reasonable plug-in estimator under the assumptions that (i) the covariate(s) that create the strata are related to the potential outcomes and variance and that (ii) the small strata are more similar to each other than the larger strata.

Covariates could also be exploited using linear regression to predict variance for the intractable blocks or estimate a variance model for all blocks in a pooled manner; see, for example, \cite{fogarty_2017}.
If we believe that the variance of the estimator in each block is related to the block size, we could fit a linear regression for the big blocks, of variance versus their size, and then extrapolate to the small blocks.
Alternatively, if nothing is known and there are very few small blocks, an average (or the largest) of the big block variance estimates might be used.

Any of these plug-in estimators can either be used for all the blocks or used to simply fill in any missing components of the small blocks.
For instance, if all of the small blocks are such that they have multiple controls but only one treated unit, we can calculate $\sckest$ as usual but approximate $\stkest$ based on one of the previously mentioned methods.

In general, many plug-ins could be appropriate, based on what assumptions the researcher is able to make.
The choice of plug-in estimator should be chosen prior to running the experiment and should be based on the researchers assumptions and knowledge at that time.
Trying several plug-in estimators and using the smallest will create bias. 

\subsection{The RCT-YES estimator}\label{subsubsec:rct_yes_est}
One might also consider an estimator suggested in the RCT-YES manuscript \citep[][p. 83]{schochet_2016}.
The form of this estimator, using block sizes as weights, is
\[\widehat{\sigma}^2_{RCT} = \frac{1}{K\left(K-1\right)(\frac{n}{K})^2}\sum_{k=1}^K\left(n_k\taukest - \frac{n}{K}\tauestbk\right)^2 . \]
As discussed and proven in the RCT-YES documentation, this estimator is consistent under an infinite population of an infinite number of strata of infinite size, where we sample strata and then units within strata.
This is the random sampling of strata setting in the remark of Section~\ref{subsec:inf_finstrat_framework}.
This estimator differs from our estimators $\varestsbm$ and $\varestsbp$ by putting the weights inside the square.
Unfortunately, moving the weighting inside the squares can cause large bias in the finite setting and the stratified sampling framework.
In fact, in the simulations comparing variance estimator performance in the finite sample, presented in Section~\ref{sec:sim}, the RCT-YES bias and variance was high enough that it was not comparable to the other estimators presented.
This estimator is targeting a superpopulation quantity, thus the standard errors are larger in part to capture the additional variation of the strata being a random sample.

We discussed the performance of the original RCT-YES estimator with Dr. Schochet (personal correspondence, April 2018), and he proposed an alternate estimator.
Again using the block sizes as weights, this estimator has the form
\[\widehat{\sigma}^2_{RCT, 2} = \frac{1}{K\left(K-1\right)(\frac{n}{K})^2}\sum_{k=1}^K\nk^2\left(\taukest - \tauestbk\right)^2 \]
and is rooted in survey sampling methods \citep{CochranWilliamG.WilliamGemmell1977St}.
This estimator is more stable because the weights are not inside the parentheses.
This estimator is still motivated by a superpopulation sampling framework, and takes the variability of the blocks into account.
Finite sample performance using this estimator on all of the blocks does not perform well, unless all blocks have the same $\tau_k$, which aligns with what we expect from explorations of our small block estimators.
If used in combination as a hybrid estimator, its performance is very similar to that of the hybrid using $\varestsbp$.

\section{Connecting to \cite{mukerjee2018using}}\label{supp:muk_et_al_paper}

\cite{mukerjee2018using} sets up a framework for representing conservative variance estimators for a very general set of causal estimands and estimators.
The big block variance estimator, $\varestbkone$, falls within their framework, as shown in their paper.
We can also show that the small block estimators $\varestsbs$ and $\varestsbm$ do as well.
However, it does not seem that our other small block estimator, $\varestsbp$, is within their framework scope.
In fact, it appears necessary to go beyond this framework, as we have done with $\varestsbp$, to obtain a conservative variance estimator for small blocks of variable size when there are not multiple blocks of the same size.
Note that the hybrid estimator corresponding to any small block estimator that can be obtained within their framework can also be obtained within their framework.
We present some details below to understand these connections more clearly.

Under the notation of \cite{mukerjee2018using}, they consider $T$ to be any partition of units into treatment groups.
In other words, it represents a single assignment of units into treatment and control.
As all of the partitions $T$ are equally probably under the block randomized assignment mechanism and it is clear that we are focusing on block randomization here, we will largely drop this notation.
We next define quantities used in \cite{mukerjee2018using} but not in our paper, but largely keep with the notation from our paper.
Equation (5) of \cite{mukerjee2018using} defines the treatment effect estimator as 
\[\widehat{\tau} = \sum_{z \in \{c,t\}}g(z)\widehat{\bar{Y}}(z), \]
where $\widehat{\bar{Y}}(z)$ is some sort of average of units under treatment $z$ and in our case $g(t) = 1$ and $g(c) = -1$.
Equation (6) defines 
\[\widehat{\bar{Y}}(z) = a(T,z) + \sum_{i: Z_i=z}b_i(T,z)Y_i(z).\]
In our case $a(T,z) = 0$ and $b_i(T,z) = \frac{\nk}{n}\frac{1}{\nzk}$ for units in block $k$ under treatment $z$ such that
\[\widehat{\bar{Y}}(z)=\sum_{k=1}^K\sum_{i: Z_i=z, b_i=k} \frac{\nk}{n}\frac{1}{\nzk}Y_i(z) = \sum_{k=1}^K \frac{\nk}{n}\meanpozkobs.\]
Then we have
\[\widehat{\tau}=\tauestbk=\sum_{k=1}^K \frac{\nk}{n}\left(\meanpotkobs-\meanpockobs\right),\]
as desired.

To layout the framework for variance estimation of \cite{mukerjee2018using}, we need to define some more quantities.
First, we have the probabilities
\begin{align*}
\pi_i(z) &= E[\mathbb{I}_{Z_i=z}] = \frac{\nzk}{\nk} \quad \text{for unit $i: b_i=k$,}\\
\pi_{ii^*}(z, z^*) &= E[\mathbb{I}_{Z_i=z}\mathbb{I}_{Z_{i^*}=z^*}] = \begin{cases}
\frac{\nzk}{\nk}\frac{n_{z^*,j}}{n_j} & \text{for unit $i : b_i=k$, $i^*: b_{i^*}=j$ with $k \neq j$,}\\
\frac{\nzk(\nzk-1)}{\nk(\nk-1)} & \text{for unit $i, i^*: b_i=b_{i^*} = k$, $z = z^*$,}\\
\frac{\ntk\nck}{\nk(\nk-1)} & \text{for unit $i, i^* : b_i=b_{i^*} = k$, $z \neq z^*$.}\\
\end{cases}
\end{align*}

Then we have the quantities from equations (12) of \cite{mukerjee2018using}, which we simplify in our setting as follows:
\begin{align*}
M &= \sum_{z\in \{c,t\}} \sum_{z^* \in \{c,t\}} g(z) g(z^*) A(z,z^*)\\
& = \sum_{z\in \{c,t\}} \sum_{z^* \in \{c,t\}} g(z) g(z^*) a(T,z)a(T,z^*)\\
& = 0,\\
M_i(z) &= g(z) \sum_{z^* \in \{c,t\}}  g(z^*) \left(A_i^{(1)}(z,z^*) + A_i^{(2)}(z^*,z)\right)\\
& = g(z) \sum_{z^* \in \{c,t\}}  g(z^*) \left(E[\mathbb{I}_{Z_i=z}]a(T,z^*)b_i(T,z) + E[\mathbb{I}_{Z_i=z^*}]a(T,z)b_i(T,z^*)\right)\\
& = 0,\\
M_{ii}(z) &= (g(z))^2B_{ii}(z, z)\\
& = (g(z))^2E[\mathbb{I}_{Z_i=z}(b_i(T,z))^2]\\
& = \frac{\nzk}{\nk}\left(\frac{\nk}{n}\frac{1}{\nzk}\right)^2 \quad \text{( for $i: b_i = k$)}\\
& = \frac{1}{n^2}\frac{\nk}{\nzk},\\
M_{ii^*}(z, z^*) &= g(z)g(z^*)B_{ii^*}(z, z^*)\\
&= g(z)g(z^*)E[\mathbb{I}_{Z_i=z}\mathbb{I}_{Z_{i^*}=z^*}b_i(T,z)b_{i^*}(T,z^*)]\\
& = \begin{cases}
g(z)g(z^*)\frac{1}{n^2} & \text{for unit $i: b_i=k$, $i^*: b_{i^*} = j$ with $k \neq j$,}\\
\frac{1}{n^2}\frac{\nk(\nzk-1)}{\nzk(\nk-1)} & \text{for unit $i, i^*: b_i = b_{i^*} = k$, $z = z^*$,}\\
-\frac{1}{n^2}\frac{\nk}{\nk-1} & \text{for unit $i, i^* : b_i = b_{i^*} = k$, $z \neq z^*$.}
\end{cases}
\end{align*}

To get the variance and variance estimators we need an additional quantity from Equation (17) of \cite{mukerjee2018using},
\begin{align*}
\tilde{M}_{ii^*}(z, z^*) &= g(z)g(z^*)\left[B_{ii^*}(z, z^*) + q_{ii^*} - \frac{1}{n^2}\right]\\
&=\begin{cases}
g(z)g(z^*)q_{ii^*} & \text{for unit $i: b_i=k$, $i^*: b_{i^*} = j$ with $k \neq j$,}\\
q_{ii^*} -\frac{1}{n^2}\frac{\nk-\nzk}{\nzk(\nk-1)}   & \text{for unit $i, i^*: b_i=b_{i^*} = k$, $z = z^*$,}\\
- \left[\frac{1}{n^2}\frac{1}{\nk-1} + q_{ii^*}\right]& \text{for unit $i, i^*:b_i = b_{i^*} = k$, $z \neq z^*$,}\\
\end{cases}
\end{align*}
with $q_{ii^*}$ being entries from a matrix, as explained below.

We can use these quantities to derive the true variance of the blocking estimator.
This was shown in \cite{mukerjee2018using} as an example of the use of their framework.
Further, those authors show that the variance can be written using some positive-semi definite $\bm{Q}$ matrix with $(i,i^*)$ entry denoted $q_{ii^*}$ such that $\bm{Q}\bm{J} = \bm{0}$, where $\bm{J}$ is a matrix of all ones, and $q_{ii} = 1/N^2$ (see equation (15), (18), (19) of \cite{mukerjee2018using}), as 
\[\text{var}(\widehat{\tau}) = V_{\bm{Q}}(\widehat{\tau}) - \bm{\tau}'\bm{Q}\bm{\tau},\]
where $\bm{\tau}' = (\tau_1,\dots,\tau_n)$ and
\begin{align*}
V_{\bm{Q}}(\widehat{\tau}) =& M + \sum_{z \in \{c,t\}}\sum_{i=1}^n\left(M_i(z)Y_i(z) + M_{ii}(z)(Y_i(z))^2\right)\\
&+\sum_{z \in \{c,t\}}\sum_{z^* \in \{c,t\}}\sum_{i=1}^n\sum_{i^* (\neq i) =1}^n\tilde{M}_{ii^*}(z,z^*)Y_i(z)Y_{i^*}(z^*).
\end{align*}
In our setting,
\begin{align*}
V_{\bm{Q}}(\tauestbk)=& \sum_{z \in \{c,t\}}\sum_{k=1}^K\sum_{i: b_i = k}\frac{1}{n^2}\frac{\nk}{\nzk}(Y_i(z))^2\\
&+ \sum_{z \in \{c,t\}}\sum_{k=1}^K\sum_{i: b_i=k}\sum_{i^*(\neq i): b_{i^*}=k} \left(q_{ii^*} -\frac{1}{n^2}\frac{\nk-\nzk}{\nzk(\nk-1)}\right)Y_i(z)Y_{i^*}(z)\\
& - 2\sum_{k=1}^K\sum_{i: b_i=k}\sum_{i^*(\neq i): b_{i^*}=k}\left(\frac{1}{n^2}\frac{1}{\nk-1} + q_{ii^*}\right)Y_i(t)Y_{i^*}(c)\\
&+\sum_{z \in \{c,t\}}\sum_{z^* \in \{c,t\}}\sum_{k=1}^K\sum_{j (\neq k)=1}^K\sum_{i:b_i=k}\sum_{i^*:b_{i^*}=j}g(z)g(z^*)q_{ii^*}Y_i(z)Y_{i^*}(z^*).
\end{align*}

Armed with this set up, the authors show that we can obtain a conservative variance estimator by estimating $V_{\bm{Q}}(\widehat{\tau})$ with (equation (21) of \cite{mukerjee2018using}),
\begin{align*}
\widehat{V}_{\bm{Q}}(\widehat{\tau}) =& M + \sum_{z\in \{c,t\}}\sum_{i:Z_i=z}\frac{1}{\pi_i(z)}\left(M_i(z)Y_i(z) + M_{ii}(z)(Y_i(z))^2\right)\\
& + \sum_{z\in \{c,t\}}\sum_{z^*\in \{c,t\}}\sum_{i: Z_i=z}\sum_{i^*(\neq i): Z_{i^*} = z^*}\frac{\tilde{M}_{ii^*}(z,z^*)}{\pi_{ii^*}(z,z^*)}Y_i(z)Y_{i^*}(z^*).
\end{align*}

In our setting, this simplifies to
\begin{align*}
\widehat{V}_{\bm{Q}}(\tauestbk)
=& \sum_{z\in \{c,t\}}\sum_{i: Z_i=z}\frac{M_{ii}(z)}{\pi_i(z)}(Y_i(z))^2 + \sum_{z\in \{c,t\}}\sum_{k=1}^K\sum_{i: b_i=k, Z_i=z}\sum_{i^*(\neq i): b_{i^*} =k, Z_{i^*}=z}\frac{\tilde{M}_{ii^*}(z,z)}{\pi_{ii^*}(z,z)}Y_i(z)Y_{i^*}(z)\\
& + 2\sum_{k=1}^K\sum_{i: b_i=k, Z_i=t}\sum_{i^*(\neq i): b_{i^*}=k, Z_{i^*}=c}\frac{\tilde{M}_{ii^*}(t,c)}{\pi_{ii^*}(t,c)}Y_i(t)Y_{i^*}(c)\\
& + \sum_{z\in \{c,t\}}\sum_{z^*\in \{c,t\}}\sum_{k=1}^K\sum_{j (\neq k)=1}^K\sum_{i: b_i=k, Z_i=z}\sum_{i^*: b_{i^*}=j, Z_{i^*}=z^*}\frac{\tilde{M}_{ii^*}(z,z^*)}{\pi_{ii^*}(z,z^*)}Y_i(z)Y_{i^*}(z^*)\\
=& \underbrace{\sum_{z\in \{c,t\}}\sum_{k=1}^K\sum_{i: b_i=k, Z_i=z}\frac{\nk^2}{n^2}\frac{1}{\nzk^2}(Y_i(z))^2}_{\textbf{A}}\\
& +\underbrace{ \sum_{z\in \{c,t\}}\sum_{k=1}^K\sum_{i: b_i=k, Z_i=z}\sum_{i^*(\neq i): b_{i^*}=k, Z_{i^*}=z}\frac{q_{ii^*} -\frac{1}{n^2}\frac{\nk-\nzk}{\nzk(\nk-1)}}{\frac{\nzk(\nzk-1)}{\nk(\nk-1)}}Y_i(z)Y_{i^*}(z)}_{\textbf{B}}\\
& - \underbrace{2\sum_{k=1}^K\sum_{i: b_i=k, Z_i=t}\sum_{i^*: b_{i^*}=k, Z_{i^*}=c}\frac{\frac{1}{n^2}\frac{1}{\nk-1} + q_{ii^*}}{\frac{\ntk\nck}{\nk(\nk-1)}}Y_i(t)Y_{i^*}(c)}_{\textbf{C}}\\
& + \underbrace{\sum_{z\in \{c,t\}}\sum_{z^*\in \{c,t\}}\sum_{k=1}^K\sum_{j(\neq k)=1}^K\sum_{i: b_i=k, Z_i=z}\sum_{i^*: b_{i^*}=j, Z_{i^*}=z^*}\frac{g(z)g(z^*)q_{ii^*}}{\frac{\nzk}{\nk}\frac{n_{z^*,j}}{n_j}}Y_i(z)Y_{i^*}(z^*)}_{\textbf{D}}.
\end{align*}
To connect this type of estimator to the estimators in our paper, we need to find specifications of the $\bm{Q}$ matrix such that the estimators are equivalent.
In particular, according to the assumption in Equation (22) of \cite{mukerjee2018using}, this estimator requires $\pi_{ii^*}(z, z^*) >0$ whenever $\tilde{M}_{ii^*}(z,z^*) \neq 0$.
To meet this requirement, we need $\tilde{M}_{ii^*}(z,z^*) = 0$ whenever $\pi_{ii^*}(z, z^*) =0$, which occurs in term \textbf{B} in the above expression if we have one treated (or one control) unit in block $k$.
This means that for any block $k$ with only one unit assigned to treatment $z$ we need
\[q_{ii^*} -\frac{1}{n^2}\frac{\nk-\nzk}{\nzk(\nk-1)} = 0  \implies q_{ii^*} =\frac{1}{n^2}\frac{\nk-\nzk}{\nzk(\nk-1)} = \frac{1}{n^2}\frac{\nk-1}{\nk-1} = \frac{1}{n^2}\]
for $i, i^*: b_i = b_{i^*} = k$.

To connect estimators of the form $\widehat{V}_{\bm{Q}}(\tauestbk)$ to our variance estimators for hybrid blocked experiments, we will put our small block estimators in the same form as  $\widehat{V}_{\bm{Q}}(\tauestbk)$ to make the comparison easier.
Our pooled small block estimator ($\varestsbp$), as well as the matched pairs estimator ($\varestsbs$), is of the following form, where $a_k$ are some weights such that $a_k-2\frac{\nk}{n}a_k + \frac{\nk^2}{n^2}\sum_{j=1}^Ka_j = \nk^2/n^2$ (see Supplementary Material~\ref{append:var_est_small_p_derv} for proof that the weights for $\varestsbs$ and $\varestsbp$ have this property):
\begin{align*}
&\sum_{k=1}^Ka_k\left(\taukest - \tauestbk\right)^2\\
 &= \sum_{k=1}^Ka_k\left(\taukest^2 - 2\tauestbk\taukest+\tauestbk^2\right)\\
&= \sum_{k=1}^Ka_k\left(\taukest^2 - 2\frac{\nk}{n}\taukest^2 - 2\taukest\sum_{j\neq k}\frac{n_j}{n}\widehat{\tau}_{j}+\sum_{h=1}^K\frac{n_h^2}{n^2}\widehat{\tau}_{h}^2+\sum_{h=1}^K\sum_{f\neq h}\frac{n_hn_f}{n^2}\widehat{\tau}_{h}\widehat{\tau}_{f}\right)\\
&= \sum_{k=1}^K\left(a_k-2\frac{\nk}{n}a_k + \frac{\nk^2}{n^2}\sum_{h=1}^Ka_h\right)\taukest^2 +\sum_{k=1}^K\sum_{j \neq k}\left(\frac{\nk n_j}{n^2}\sum_{h=1}^Ka_h - 2a_k\frac{n_j}{n}\right)\taukest\widehat{\tau}_{j}\\
&= \sum_{k=1}^K\left(\frac{\nk^2}{n^2}\right)\Bigg(\sum_{z \in \{c,t\}}\sum_{i: b_i=k, Z_i=z}\frac{1}{\nzk^2}\left(Y_i(z)\right)^2  +\sum_{z \in \{c,t\}}\sum_{i: b_i = k, Z_i=z}\sum_{i^* (\neq i): b_{i^*} = k, Z_{i^*}=z}\frac{1}{\nzk^2}Y_i(z)Y_{i^*}(z)\\
& \qquad \qquad \qquad \qquad  -2\sum_{i:b_i=k, Z_i=t}\sum_{i^*: b_{i^*}=k, Z_{i^*}=c}\frac{1}{\ntk\nck}Y_i(t)Y_{i^*}(c)\Bigg)\\
& \quad+\sum_{k=1}^K\sum_{j \neq k}\left(\frac{\nk n_j}{n^2}\sum_{h=1}^Ka_h - 2a_k\frac{n_j}{n}\right)\sum_{z\in \{c,t\}}\sum_{z^*\in \{c,t\}}\sum_{i: b_i=k, Z_i=z}\sum_{i^*: b_{i^*}=j, Z_{i^*}=z^*}\frac{g(z)g(z^*)}{\nzk n_{z^*,j}}Y_i(z)Y_{i^*}(z^*)\\
&= \underbrace{\sum_{k=1}^K\frac{\nk^2}{n^2}\sum_{z \in \{c,t\}}\sum_{i:b_i=k, Z_i=z}\frac{1}{\nzk^2}\left(Y_i(z)\right)^2}_{\textbf{A}_1}+\underbrace{\sum_{k=1}^K\frac{\nk^2}{n^2}\sum_{z \in \{c,t\}}\sum_{i: b_i=k, Z_i=z}\sum_{i^* (\neq i):b_{i^*}=k, Z_{i^*}=z}\frac{1}{\nzk^2}Y_i(z)Y_{i^*}(z)}_{\textbf{B}_1}\\
&\quad -\underbrace{2\sum_{k=1}^K\frac{\nk^2}{n^2}\sum_{i: b_i=k, Z_i=t}\sum_{i^*: b_{i^*}=k, Z_{i^*}=c}\frac{1}{\ntk\nck}Y_i(t)Y_{i^*}(c)}_{\textbf{C}_1}\\
& \quad+\underbrace{\sum_{k=1}^K\sum_{j \neq k}\left(\frac{\nk n_j}{n^2}\sum_{h=1}^Ka_h - 2a_k\frac{n_j}{n}\right)\sum_{z\in \{c,t\}}\sum_{z^*\in \{c,t\}}\sum_{i: b_i=k, Z_i=z}\sum_{i^*: b_{i^*}=j, Z_{i^*}=z^*}\frac{g(z)g(z^*)}{\nzk n_{z^*,j}}Y_i(z)Y_{i^*}(z^*)}_{\textbf{D}_1}.\\
\end{align*}

We see that $\textbf{A}_1$ already matches $\textbf{A}$ from $\widehat{V}_{\bm{Q}}(\tauestbk)$.
We can simplify \textbf{B} and \textbf{C} using $q_{ii^*} = 1/n^2$ for $i, i^*$ in block $k$ such that there is only one unit assigned to $z$.

Starting with term $\textbf{B}$, note that if we assume each block is small (has one treated or one control) unit then for one $z$ term \textbf{B} is 0 and for the other we have, for units in block $k$, constant
\begin{align*}
\frac{q_{ii^*} -\frac{1}{n^2}\frac{\nk-\nzk}{\nzk(\nk-1)}}{\frac{\nzk(\nzk-1)}{\nk(\nk-1)}}&=\frac{\frac{1}{n^2} -\frac{1}{n^2}\frac{\nk-\nzk}{\nzk(\nk-1)}}{\frac{\nzk(\nzk-1)}{\nk(\nk-1)}}\\
&=\frac{1}{n^2}\frac{\nzk(\nk-1) -\nk+\nzk}{\frac{\nzk^2(\nzk-1)}{\nk}}\\
&=\frac{1}{n^2}\frac{\nk(\nzk-1)}{\frac{\nzk^2(\nzk-1)}{\nk}}\\
&=\frac{\nk^2}{n^2}\frac{1}{\nzk^2}.
\end{align*}
So we have that $\textbf{B}_1$ matches \textbf{B}.

Now for \textbf{C},
\begin{align*}
\frac{\frac{1}{n^2}\frac{1}{\nk-1} + q_{ii^*}}{\frac{\ntk\nck}{\nk(\nk-1)}} &= \frac{\frac{1}{n^2}\frac{1}{\nk-1} + \frac{1}{n^2}}{\frac{\ntk\nck}{\nk(\nk-1)}}\\
 &=\frac{1}{n^2} \frac{\frac{\nk}{\nk-1}}{\frac{\ntk\nck}{\nk(\nk-1)}}\\
  &=\frac{\nk^2}{n^2} \frac{1}{\ntk\nck}.
\end{align*}
Hence we have $\textbf{C}_1$ matches \textbf{C}.

Now we have the final term, in which we need to specify $q_{ii^*}$ for $i$ and $i^*$ in different blocks.
Let's start by simplifying the constant for $\textbf{D}_1$.
Now we will have to plug in the $a_k$ for $\varestsbp$.
\begin{align*}
\frac{\nk n_j}{n^2}\sum_{h=1}^Ka_h - 2a_k\frac{n_j}{n} =&\frac{\nk n_j}{n^2}\left(\sum_{h=1}^Ka_h - 2\frac{a_k}{\nk}n\right)\\
=&\frac{\nk n_j}{n^2}\left(\frac{\sum_{h=1}^K\frac{n_h^2}{n-2n_h}}{n+\sum_{h=1}^K\frac{n_h^2}{n-2n_h}} - \frac{2\nk n}{(n-2\nk)\left(n+\sum_{h=1}^K\frac{n_h^2}{n-2n_h}\right)}\right)\\
=&\frac{\nk n_j}{n^2}\left(\frac{(n-2\nk)\sum_{h=1}^K\frac{n_h^2}{n-2n_h}-2\nk n}{(n-2\nk)\left(n+\sum_{h=1}^K\frac{n_h^2}{n-2n_h}\right)}\right)\\
=&\frac{\nk n_j}{n^2}\left(\frac{n\sum_{h=1}^K\frac{n_h^2}{n-2n_h}-2\nk\left(n+\sum_{h=1}^K\frac{n_h^2}{n-2n_h}\right)}{(n-2\nk)\left(n+\sum_{h=1}^K\frac{n_h^2}{n-2n_h}\right)}\right)\\
=&\frac{\nk n_j}{n^2}\left(\frac{n\left(n+\sum_{h=1}^K\frac{n_h^2}{n-2n_h}\right)-n^2-2\nk\left(n+\sum_{h=1}^K\frac{n_h^2}{n-2n_h}\right)}{(n-2\nk)\left(n+\sum_{h=1}^K\frac{n_h^2}{n-2n_h}\right)}\right)\\
=&\frac{\nk n_j}{n^2}\left(1-\frac{n^2}{(n-2\nk)\left(n+\sum_{h=1}^K\frac{n_h^2}{n-2n_h}\right)}\right)
\end{align*}
So, comparing to term \textbf{D}, we need
\[q_{ii^*} = \frac{1}{n^2}\left(1-\frac{n^2}{(n-2\nk)\left(n+\sum_{h=1}^K\frac{n_h^2}{n-2n_h}\right)}\right).\]

First note that if all blocks are of the same size, R, in which case $\varestsbp$ reduces down to $\varestsbs$, then
\begin{align*}
\frac{1}{n^2}\left(1-\frac{n^2}{(n-2\nk)\left(n+\sum_{h=1}^K\frac{n_h^2}{n-2n_h}\right)}\right) =&\frac{1}{R^2K^2}\left(1-\frac{R^2K^2}{(RK-2R)\left(RK+\sum_{h=1}^K\frac{R^2}{RK-2R}\right)}\right)\\
    =&\frac{1}{R^2K^2}\left(1-\frac{K}{K-1}\right)\\ 
    =&-\frac{1}{R^2K^2(K-1)}.
\end{align*}
This is actually essentially the same solution as example (b) in Section 5.1 of \cite{mukerjee2018using} for the split plot example with plots of the same size.
We can do a check that this works according to their criteria.
We need $\bm{Q}\bm{J} = \bm{0}$.
Let's take the $i$th row and arbitrary $j$th entry of $\bm{Q}\bm{J}$.
For all units belonging to the same block as unit $i$, $q_{ii^*} = 1/n^2$. Otherwise, $q_{ii^*} = 1/(n^2(K-1))$.
\begin{align*}
[\bm{Q}\bm{J}]_{ij} =& \frac{R}{n^2} - \frac{R(K-1)}{R^2K^2(K-1)}\\
=& \frac{R}{R^2K^2} - \frac{1}{RK^2}\\
=& 0
\end{align*}
So yes: for multiple blocks of the same size, the standard matched pairs variance estimator, $\varestsbs$, falls within the \cite{mukerjee2018using} framework for variance estimators.

This means that we can also do a block diagonalized version of this $\bm{Q}$ matrix to get our non-pooled small block estimator, $\varestsbm$, where we split blocks up by size first and then block diagonalize according to these splits.
More generally, estimators will be block diagonalized according to whether the pieces of the variance estimators include terms from multiple blocks.
For instance, with the big block part of the estimator, we would have all of the big blocks being separate block diagonals because we estimate variance within each block separately.
Similarly, for $\varestsbm$ because we estimate the variance within each small block group of the same size separately, we would block diagonalize according to block size for that piece.
On the other hand, if we had all blocks of the same size then there is no block diagonal because $\varestsbs$ uses all of the blocks together to estimate the overall variance.

Now what about $\varestsbp$ for small blocks of varying size?
The $q_{ii^*}$ terms for $i: b_i=k$ and $i^*: b_{i^*} = j$ with $k \neq j$ would be,
\[q_{ii^*} = \frac{1}{n^2}\left(1-\frac{n^2}{(n-2\nk)\left(n+\sum_{h=1}^K\frac{n_h^2}{n-2n_h}\right)}\right),\]
which are not guaranteed to be symmetric in this case (because of the $\nk$ term in the denominator).
So the $\bm{Q}$ matrix would not fall into \cite{mukerjee2018using} framework because it is not positive semi-definite.
Note that we assure positivity of bias for our estimator by requiring that $\nk < n/2$ for all blocks when using the pooled estimator.
Therefore, our estimator does not seem to fall within their framework and it appears to us that extending beyond their framework is necessary to solve the problem of small blocks of variable size.

\section{Data generating process for simulations}\label{append:more_sim}

The data generating process used for the simulations comparing the variance estimators gives us a single finite data set.
We then repeatedly randomize units to treatment, according to blocked randomization, to assess finite sample behavior.

The initial potential outcomes for the units in each block were drawn from a bivariate normal distribution, with the means and covariance matrix as follows (shown for a unit in block $k$):
\[\begin{pmatrix}
Y_i(0)\\
Y_i(1)
\end{pmatrix}\sim MVN\left(
\begin{pmatrix}
\alpha_k\\
\alpha_k + \beta_k
\end{pmatrix},\begin{pmatrix}
1 & \rho\\
\rho &1
\end{pmatrix}\right).
\]

The correlation of potential outcomes, $\rho$, was varied among 0, 0.5, and 1.
We controlled how differentiated the blocks were, and how heterogeneous the treatment effects across blocks were, by varying $\alpha_k$ and $\beta_k$.
We set $\alpha_k$ as $\alpha_k = \Phi^{-1}\left(1-\frac{k}{K+1}\right)a$. 
Similarly, $\beta_k = 5 + \Phi^{-1}\left(1-\frac{k}{k+1}\right)b$. 
The larger the $a$, the more the mean control potential outcomes for the blocks were spread apart.
The larger the $b$, the more heterogeneous the treatment impacts.
The parameters $a$ and $b$ were each varied among the values (0,0.1,0.3,0.5,0.8,1,1.5,2).
We keep the number and sizes of blocks fixed.
The blocks were ordered by size with the smallest block as block one.
As a consequence, the smaller blocks have both larger means under control and larger average treatment effects.

Simulations were run over assignment of units to treatments under a blocked design, which was done 5000 times for each combination of factors.
Before evaluating our overall results, we first checked that the simulation values agreed with the biases calculated for our variance estimators via the bias formulas in Section~\ref{subsec:var_est_perf_fin} of our paper.
They did.
Replication code showing this is available.

\section{Simulation investigating the variance of the variance estimators}\label{append:var_est_var}

In this auxiliary simulation we examine how the blocking variance estimators compare amongst themselves in terms of their own variance.
To assess this, we examine the variance of our variance estimators from our simulation study in Section~\ref{sec:sim} with the data generating process given in Section~\ref{append:more_sim}.

Results are on Figure~\ref{plot:var_var}.
We see that, in terms of variance, the estimators are generally comparable.
We expect more instability from estimators that utilize only information from the estimated average treatment effects, not from the variation of the individual units.
We see that the weighted regression estimator has the lowest variability.

\begin{figure}
\centering
\includegraphics[width=0.8\textwidth]{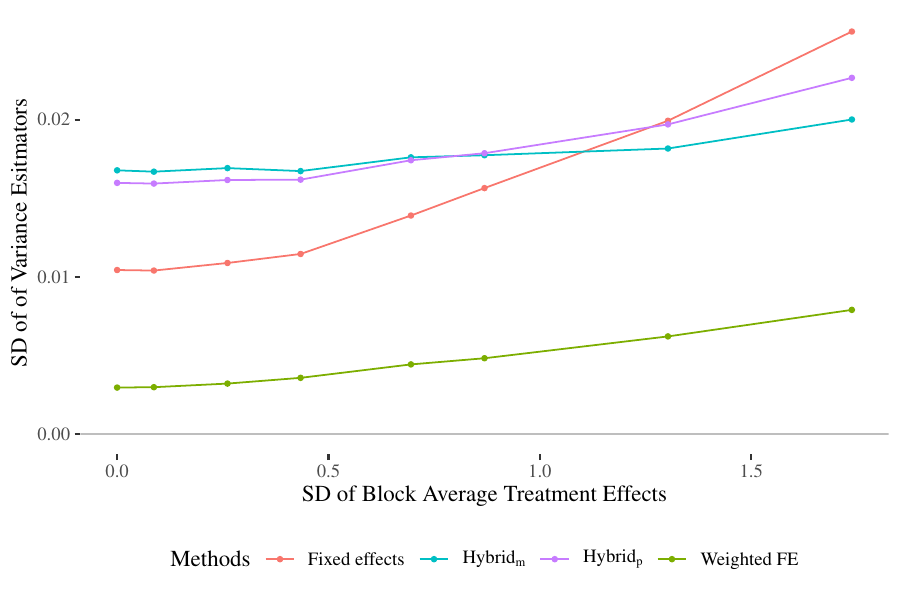}
\caption{Simulations to assess variance estimators' variance.
The x-axis shows the standard deviation of block average treatment effects.
Points show the average across the values of $\rho$ and the standard deviation of block average control potential outcomes in the simulation. (The trends for different $\rho$ were essentially the same.) FE stands for fixed effects.}
\label{plot:var_var}
\end{figure}

%%
%% Creation of small block variance estimator
%% 

\section{Creation and bias of $\varestsbm$}\label{append:var_est_small_m}

We first construct the variance estimator $\varestsbm$, and then derive the bias for this estimator for different frameworks.

\subsection{Creation of $\varestsbm$, Equation~(\ref{eq:var_est_sb_m})}\label{append:var_est_sb_m_derv}
To formally state the method described by Equation~(\ref{eq:var_est_sb_m}) from Section~\ref{subsec:var_est_small}, we first express our estimands and estimators in terms of weighted averages of estimates within collections of same-size blocks.
In the following, let there be $J$ unique block sizes in the sample. 
Let $\numsizej$ be the $j$th block size and let $\numsb$ be the number of blocks in the population of size $\numsizej$.
So then $n = \sum_{j=1}^J\numsizej\numsb$.
In particular, the sample average treatment effect for all units in blocks of size $\numsizej$ is
\[\tausbfsj = \frac{1}{\numsb}\sum_{k: \nk = \numsizej} \taukfs.\]
Let $N_j = \numsizej \numsb$ be the total number of units in the small blocks of size $\numsizej$.
Then the overall sample average treatment effect in terms of these $\tausbfsj$ is
\begin{equation}
\tausbfs = \frac{1}{\sum_{i=1}^J N_j}\sum_{j=1}^J N_j \tausbfsj. \label{eq:tau_small_j}	
\end{equation}
$\tausbfs$ is the same as $\taufs$ as before; we add the subscript ``small'' here to clarify the notation when we discuss hybrid experiments.
Note that these definitions are analogous for the infinite population, which are indicated by removing the $\mathcal{S}$ subscript.

The treatment effect estimators can be written in analogous form to the above.
We have unbiased estimators $\tausbjest = \frac{1}{\numsb}\sum_{k: \nk = \numsizej} \taukest$ for the average treatment effects in blocks of size $\numsizej$.
Simply plug them into Equation~\ref{eq:tau_small_j} to obtain an overall treatment effect estimator.

As discussed in Section~\ref{subsec:var_est_small}, within each piece $j$, use a variance estimator with the same form as Equation~\ref{eq:var_est_sb_same_size}:
\begin{align*}
\varestsbj  = \frac{1}{\numsb(\numsb-1)}\sum_{k: \nk = \numsizej} (\taukest - \tausbjest)^{2}
\end{align*}

Then combine to create an overall variance estimator:
\begin{align*}
\varestsbm = \frac{1}{\left(\sum_{j=1}^J N_j \right)^2}\sum_{j=1}^J N_j^2\varestsbj. 
\end{align*}

\subsection{Proof of Corollary~\ref{lemma:var_est_sb_m_fin_samp} (Bias under finite sample) and Corollary~\ref{lemma:bias_var_est_sb_m_strat_samp} (Bias under stratified sampling)}\label{append:bias_var_est_sb_m_strat_samp}

\begin{proof}
For this section assume that we are in the finite sample framework. 
The results for the stratified sampling from an infinite population framework follow directly by changing the expectations and notation.

First we will focus on $\varestsbj$ which is the variance estimator for $\tausbjest$. 
Note that
\begin{align*}
\text{var}\left(\tausbjest|\mathcal{S}\right) &= \text{var}\left(\frac{1}{\numsb}\sum_{k: \nk = \numsizej} \taukest\Big|\mathcal{S}\right)= \frac{1}{\numsb^2}\sum_{k: \nk = \numsizej}\text{var}\left( \taukest|\mathcal{S}\right).
\end{align*}

\begin{align*}
\EE&\left[\varestsbj|\mathcal{S}\right]\\
 &=\EE\left[ \frac{1}{\numsb(\numsb-1)}\sum_{k: \nk = \numsizej} (\taukest - \tausbjest)^{2}\Big|\mathcal{S}\right]\\
&=\frac{1}{\numsb(\numsb-1)}\EE\left[\sum_{k: \nk = \numsizej} \left(\taukest^2 -2\taukest\tausbjest+ \tausbjest^2\right)\Big|\mathcal{S}\right]\\
&=\frac{1}{\numsb(\numsb-1)}\left(\sum_{k: \nk = \numsizej}\left[\text{var}\left(\taukest|\mathcal{S}\right) + \taukfs^2\right] - \numsb\left[\text{var}\left(\tausbjest|\mathcal{S}\right) + \tausbfsj^2\right]\right)\\
&=\frac{1}{\numsb(\numsb-1)}\left(\sum_{k: \nk = \numsizej}\left[\frac{\numsb - 1}{\numsb}\text{var}\left(\taukest|\mathcal{S}\right) + \taukfs^2\right] - \numsb \tausbfsj\right)\\
&=\frac{1}{\numsb^2}\sum_{k: \nk = \numsizej}\text{var}\left(\taukest|\mathcal{S}\right) + \frac{1}{\numsb(\numsb-1)}\sum_{k: \nk = \numsizej}\left(\taukfs - \tausbfsj\right)^2\\
&=\text{var}\left(\tausbjest|\mathcal{S}\right) + \frac{1}{\numsb(\numsb-1)}\sum_{k: \nk = \numsizej}\left(\taukfs - \tausbfsj\right)^2.
\end{align*}

So the bias is
\[\EE\left[\varestsbj|\mathcal{S}\right] - \text{var}\left(\tausbjest|\mathcal{S}\right)=\frac{1}{\numsb(\numsb-1)}\sum_{k: \nk = \numsizej}\left(\taukfs - \tausbfsj\right)^2.\]

Now we move our attention to $\varestsbm$ which is a variance estimator for $\tausbest$.

We have
\begin{align*}
\text{var}\left(\tausbest|\mathcal{S}\right) &= \text{var}\left(\frac{1}{\sum_{i=1}^J m_iK_i}\sum_{j=1}^J\numsizej\numsb\tausbjest|\mathcal{S}\right)\\
&= \frac{1}{\left(\sum_{i=1}^J m_iK_i\right)^2}\sum_{j=1}^J\left(\numsizej\numsb\right)^2\text{var}\left(\tausbjest|\mathcal{S}\right).
\end{align*}
So then
\begin{align*}
\EE&\left[\varestsbm|\mathcal{S}\right]\\
&=\frac{1}{\left(\sum_{i=1}^J m_iK_i\right)^2}\sum_{j=1}^J\left(\numsizej\numsb\right)^2\EE\left[\varestsbj|\mathcal{S}\right]\\
&=\text{var}\left(\tausbest|\mathcal{S}\right) + \sum_{j=1}^J\frac{\numsizej^2\numsb}{\left(\sum_{i=1}^J m_iK_i\right)^2(\numsb-1)}\sum_{k: \nk = \numsizej}\left(\taukfs - \tausbfsj\right)^2.
\end{align*}

So the bias is
\[\EE\left[\varestsbm|\mathcal{S}\right] - \text{var}\left(\tausbest|\mathcal{S}\right) = \sum_{j=1}^J\frac{\numsizej^2\numsb}{\left(\sum_{i=1}^J m_iK_i\right)^2(\numsb-1)}\sum_{k: \nk = \numsizej}\left(\taukfs - \tausbfsj\right)^2.\]
\end{proof}

\subsection{Proof of Corollary~\ref{lem:rand_samp_strat_varsbm} (Bias under random sampling of strata, conditioning on block size)}\label{append:bias_var_est_sb_m_inf_strat}

\begin{proof}
Assume that we are in the random sampling of strata framework of Section~\ref{subsec:inf_finstrat_framework}.
We are focusing in on just the set of strata of size $m_j$.
We can either consider that there only exist strata of this size or we can imagine a sampling mechanism that draws these strata independently from strata of other sizes (e.g. stratified sampling by strata size), which is the same as conditioning on the number of strata of each size in the sample.
Also,
\begin{align*}
\text{var}\left(\tausbjest|\model_2\right) &= \EE\left[\text{var}\left(\tausbjest|\mathcal{S}\right)|\model_2\right] + \text{var}\left(\EE\left[\tausbjest|\mathcal{S}\right]|\model_2\right).
\end{align*}

From Appendix \ref{append:bias_var_est_sb_m_strat_samp}, we can see that
\begin{align*}
\EE&\left[\varestsbj|\mathcal{S}\right] &=\text{var}\left(\tausbjest|\mathcal{S}\right) + \frac{1}{\numsb(\numsb-1)}\sum_{k: \nk = \numsizej}\left(\taukfs - \tausbfsj\right)^2.
\end{align*}

From standard results from sampling theory (see \cite[Chapter~2]{lohr2010sampling}) we have
\[\EE\left[\frac{1}{\numsb(\numsb-1)}\sum_{k: \nk = \numsizej}\left(\taukfs - \tausbfsj\right)^2|\model_2\right] = \frac{\sigma^2_\tau}{K_j} = \text{var}\left(\tausbjest|\model_2\right).\]

Hence, we end up with
\begin{align*}
\EE\left[\varestsbj|\model_2\right] &=\EE\left[\text{var}\left(\tausbjest|\mathcal{S}\right)|\model_2\right] + \text{var}\left(\EE\left[\tausbjest|\mathcal{S}\right]|\model_2\right).
\end{align*}

Thus, this is an unbiased variance estimator in this setting.

If we have varying size but condition on the number of strata of each size that we sample, then we use the fact that stratified sampling causes the estimators for each strata size to be independent.

The proof of this result for an infinite number of infinite size strata is direct by replacing the conditioning on $\mathcal{S}$ by conditioning on $\bm{B}$ and using results from the stratified sampling framework.
\end{proof}

\section{Creation and bias of $\varestsbp$}\label{append:var_est_small_p_derv}

We next construct the variance estimator $\varestsbp$ that can estimate variance across a heterogeneous assortment of blocks, and then derive its bias under different frameworks.

\subsection{Proof of Theorem~\ref{lemma:var_est_sb_p_fin_samp} (Bias under finite sample) and Corollary~\ref{lemma:bias_var_est_sb_p_strat_samp} (Bias under stratified sampling)}\label{append:var_est_small_p_proofs}
\begin{proof}
For this section assume that we are in the finite sample framework.
The results for the stratified sampling framework follow directly by changing what we are taking the expectation with respect to and exchanging notation.
We explain how $\varestsbp$  was derived which also provides the bias of the estimator in these two frameworks.

To begin we consider, for an experiment with all small blocks, a variance estimator of the form
\[ X \equiv \sum_{k=1}^K a_k \left(\taukest - \tauestbk\right)^2 \]
for some collection of $a_k$.
We then wish to find non-negative $a_k$'s that would make this estimator as close to unbiased as possible.
In particular, we aim to create an estimator with similar bias to $\varestsbm$ but that allows for blocks of varying size.
That is, we are looking to create a similarly conservative estimator that is unbiased when the average treatment effect is constant across blocks.
This also means that we are creating the minimally biased conservative estimator of this form, without further assumptions.

The expected value of an estimator of this form is
 \begin{align*}
 \EE&\left[ X |\mathcal{S}\right]\\
  =&   \EE\left[\sum_{k=1}^{K}a_{k}\left(\taukest - \taukfs + \taukfs - \taufs + \taufs - \tauestbk\right)^{2}|\mathcal{S}\right]\\
 =&   \EE\Big[\sum_{k=1}^{K}a_{k}\Big(\left(\taukest - \taukfs\right)^{2} + \left(\taukfs - \taufs\right)^{2} + \left(\taufs - \tauestbk\right)^{2}\\
 &+ 2\left(\taukest - \taukfs\right)\left(\taukfs - \taufs\right) + 2\left(\taukest - \taukfs\right)\left(\taufs - \tauestbk\right) + 2\left(\taukfs - \taufs\right)\left(\taufs - \tauestbk\right)\Big)|\mathcal{S}\Big]\\
 =&  \sum_{k=1}^{K}a_{k}\left(\text{var}\left(\taukest|\mathcal{S}\right) +  \EE\left[\left(\taukfs \!-\! \taufs\right)^{2}|\mathcal{S}\right] +  \underbrace{\EE\left[\left(\taufs\! -\! \tauestbk\right)^{2}|\mathcal{S}\right]}_{\textbf{A}} + 2 \underbrace{\EE\left[\left(\taukest \! - \! \taukfs\right)\left(\taufs \!-\! \tauestbk\right)|\mathcal{S}\right]}_{\textbf{B}}\right)
 \end{align*}

For \textbf{A}: 
  \begin{align*}
 \EE\left[(\taufs - \tauestbk)^{2}|\mathcal{S}\right] &=  \text{var}\left(\tauestbk|\mathcal{S}\right)=\sum_{k=1}^{K}\frac{\nk^{2}}{n^{2}}\text{var}\left(\taukest|\mathcal{S}\right)
 \end{align*}
 
 For \textbf{B}:
 \begin{align*}
 \EE\left[(\taukest - \taukfs)(\taufs - \tauestbk)|\mathcal{S}\right] &= \EE\left[\left(\taukest - \taukfs\right)\sum_{j=1}^{K}\frac{n_{j}}{n}\left(\tau_{j, \mathcal{S}} - \widehat{\tau}_{j}\right)\Big|\mathcal{S}\right]\\
 &= \EE\left[-\frac{\nk}{n}\left(\taukest - \taukfs\right)^{2} +\left(\taukest - \taukfs\right)\sum_{j \neq k}\frac{n_{j}}{n}\left(\tau_{j, \mathcal{S}} - \widehat{\tau}_{j, \mathcal{S}}\right)\Big|\mathcal{S}\right]\\
  &=-\frac{\nk}{n}\text{var}\left(\taukest|\mathcal{S}\right)
 \end{align*}

Due to the assignment mechanism, $\widehat{\tau}_{j}$ will be independent of $\widehat{\tau}_{k}$ so the cross terms are all zero in the above equation.
 
 Putting \textbf{A} and \textbf{B} together, we get
  \begin{align*}
\EE\left[X |\mathcal{S}\right] &= \sum_{k=1}^{K}a_{k}\text{var}\left(\taukest|\mathcal{S}\right) + \sum_{k=1}^{K}a_{k}\left(\taukfs - \taufs\right)^{2} + \sum_{k=1}^{K}a_{k}\sum_{j=1}^{K}\frac{n_{j}^{2}}{n^{2}}\text{var}\left(\widehat{\tau}_{j}|\mathcal{S}\right)\\
& \qquad - 2\sum_{k=1}^{K}a_{k}\frac{\nk}{n}\text{var}\left(\taukest|\mathcal{S}\right)\\
 &=  \sum_{k=1}^{K}\left(a_{k} - 2a_{k}\frac{\nk}{n} + \frac{n_{k}^{2}}{n^{2}}\sum_{j=1}^{K}a_{j}\right)\text{var}\left(\taukest|\mathcal{S}\right) + \sum_{k=1}^{K}a_{k}\left(\taukfs - \taufs\right)^{2}
 \end{align*}

We now select $a_k$ to make the above as close to the true variance as possible.
The second term will be small if the $\taukfs$ do not vary much.
But this is unknown and thus we cannot select universal $a_k$ to control it.
We would like
 \[a_{k} - 2a_{k}\frac{n_{k}}{n} + \frac{n_{k}^{2}}{n^{2}}\sum_{j=1}^{K}a_{j} = \frac{n_{k}^{2}}{n^{2}}\]
 so that the first term is the true variance.
Then the second term, the bias, would be similar to that of the standard matched pairs variance estimator.
 
If we solve the $a_{k}$ as above we will obtain a conservative estimator that is unbiased when we have equal average treatment effect for all blocks, for the stratified sampling or finite framework.
To show this, consider the bias:
 \begin{align}
 \EE\left[X|\mathcal{S}\right] &- \text{var}\left(\tauestbk|\mathcal{S}\right)\nonumber\\
  &=  \sum_{k=1}^{K}\left(a_{k} - 2a_{k}\frac{\nk}{n} + \frac{\nk^{2}}{n^{2}}\sum_{j=1}^{K}a_{j} - \frac{\nk^2}{n^2}\right)\text{var}\left(\taukest|\mathcal{S}\right) + \sum_{k=1}^{K}a_{k}\left(\taukfs - \taufs\right)^{2}.\label{eq:ex_var_sp}
 \end{align}

We know $\text{var}\left(\taukest|\mathcal{S}\right) \geq 0$ for all $k$.
We also have $\sum_{k=1}^{K}a_{k}\left(\taukfs - \taufs\right)^{2} \geq 0$ so at a minimum it is 0.
This implies that to always be conservative, the first term in the above expression must always be at least 0.
Hence, to minimize the bias but remain conservative, we set $a_{k} - 2a_{k}\frac{\nk}{n} + \frac{\nk^{2}}{n^{2}}\sum_{j=1}^{K}a_{j} - \frac{\nk^2}{n^2} = 0$.
Note that $\taukfs$ and $\taufs$ are unknown and so we cannot optimize with respect to them.

Denote $C = \sum_{k=1}^{K}a_{k}$.
 Then we want to solve
 
 \begin{align*}
 a_k - 2a_k\frac{\nk}{n} + \frac{\nk^2}{n^2}C &= \frac{\nk^2}{n^2}\\
 a_k(1 - 2\frac{\nk}{n}) &= \frac{\nk^2}{n^2}(1-C)\\
 a_k &= \frac{\nk^2}{n}\frac{1-C}{n - 2\nk}
 \end{align*}
 
 But then
 \begin{align*}
 C&= \sum_{k=1}^{K} a_k = \sum_{k=1}^{K} \frac{\nk^2(1-C)}{n(n-2\nk)}\\
 \left(1 + \frac{1}{n}\sum_{k=1}^{K} \frac{\nk^2}{n-2\nk}\right)C &= \frac{1}{n}\sum_{k=1}^{K} \frac{\nk^2}{n-2\nk}
 \end{align*}
 \begin{align*}
 C&=\frac{1}{n}\sum_{k=1}^{K}\frac{\nk^2}{n-2\nk}\left(\frac{1}{1+ \frac{1}{n}\sum_{j=1}^{K}\frac{n_j^2}{n-2n_j}}\right)
 = \frac{\sum_{k=1}^{K}\frac{\nk^2}{n-2\nk}}{n+\sum_{j=1}^{K}\frac{n_j^2}{n-2n_j}}
 \end{align*}
 
 Then we have
 \begin{align*}
 a_k &= \frac{\nk^2}{n}\left(\frac{1-C}{n - 2\nk}\right)
 = \frac{\nk^2}{(n-2\nk)(n+\sum_{j=1}^K\frac{n_j^2}{n-2n_j})}.
 \end{align*}
 
 So then 
 \[\sum_{k=1}^{K} \frac{\nk^2}{\left(n-2\nk\right)\left(n+\sum_{j=1}^K\frac{n_j^2}{n-2n_j}\right)}\left(\taukest - \tauestbk\right)^{2}\]
  has bias 
  \[\sum_{k=1}^{K} \frac{\nk^2}{\left(n-2\nk\right)\left(n+\sum_{j=1}^K\frac{n_j^2}{n-2n_j}\right)}\left(\taukfs - \taufs\right)^{2}.\]
Bigger strata get weighted more heavily.

As a check, in the case where the $n_k$ are all the same, so that $n = K\nk$, the above boils down to $a_k = \frac{1}{K(K-1)}$ and $C = \frac{1}{K-1}$, giving us the classic matched pairs variance estimator.
\end{proof}

\subsection{Proof of Theorem~\ref{theorem:var_sbp_rand_samp} (Unbiasedness of $\varestsbp$ given independence of block sizes and effects)}\label{append:varsbp_rand_samp}
\begin{proof}
We start from the Equation~\ref{eq:ex_var_sp}.
As in Section~\ref{append:var_est_small_p_derv}, let 
\[ X \equiv \sum_{k=1}^K a_k \left(\taukest - \tauestbk\right)^2. \]
Then
 \begin{align*}
 \EE\left[X|\mathcal{S}\right]
  &=  \sum_{k=1}^{K}\left(a_{k} - 2a_{k}\frac{\nk}{n} + \frac{\nk^{2}}{n^{2}}\sum_{j=1}^{K}a_{j}\right)\text{var}\left(\taukest|\mathcal{S}\right) + \sum_{k=1}^{K}a_{k}\left(\taukfs - \taufs\right)^{2}.
 \end{align*}
 
 Previously we were concerned with getting the first term correct.
 But in the Random Sampling of Strata setting, the second term is trickier.
 This is especially the case if we have large blocks and thus can estimate the first term.
 So first we focus on the second term.
 Keeping the variance decomposition in mind, we ultimately want the expectation of this second term to look like $\text{var}(\taufs)$.

 As a reminder, the variance decomposition is
 \begin{align*}
 \text{var}(\tauestbk|\model_2) &=\EE\left[\text{var}(\tauestbk|\mathcal{S})|\model_2\right]+\text{var}\left(\EE[\tauestbk|\mathcal{S}]|\model_2\right)\\
 &=\EE\left[\text{var}(\tauestbk|\mathcal{S})|\model_2\right]+\text{var}\left(\taufs|\model_2\right).\\
 \end{align*}

 We assume, for simplicity, that the block sizes are independent from the treatment effects.
 This implies that $\EE[\taukfs|\model_2] = \tau$.
 
 The expected value of the second term in this setting is
 \begin{align*}
 \EE&\left[ \sum_{k=1}^{K}a_{k}\left(\taukfs - \taufs\right)^{2}\Big|\model_2 \right]\\
    =& \EE\left[ \sum_{k=1}^{K}a_{k}\taukfs^2 -2\taufs\sum_{k=1}^{K}a_k\taukfs +\taufs^2\sum_{k=1}^{K}a_{k}\Big|\model_2\right]\\
    =& \EE\Bigg[ \sum_{k=1}^{K}a_{k}\taukfs^2 -2\left(\sum_{k=1}^{K}a_k\frac{\nk}{n}\taukfs^2 +\sum_{k=1}^K\sum_{j\neq k}a_k\frac{n_j}{n}\taukfs\tau_{j,\mathcal{S}}\right)\\
    &+\left(\sum_{k=1}^{K}\frac{\nk^2}{n^2}\taukfs^2 +\sum_{k=1}^K\sum_{j\neq k}\frac{\nk n_j}{n^2}\taukfs\tau_{j,\mathcal{S}}\right)\sum_{k=1}^{K}a_{k}\Big|\model_2\Bigg]\\
        =& \EE\left[ \sum_{k=1}^{K}\left(a_{k}-2a_k\frac{\nk}{n}+\frac{\nk^2}{n^2}\sum_{i=1}^{K}a_{i}\right)\taukfs^2 -\sum_{k=1}^K\sum_{j\neq k}\left(2a_k\frac{n_j}{n}-\frac{\nk n_j}{n^2}\sum_{i=1}^{K}a_{i}\right)\taukfs\tau_{j,\mathcal{S}}\Big|\model_2\right]\\
        =& \EE\left[ \sum_{k=1}^{K}\left(a_{k}-2a_k\frac{\nk}{n}+\frac{\nk^2}{n^2}\sum_{i=1}^{K}a_{i}\right)\taukfs^2\Big|\model_2\right] -\EE\left[ \sum_{k=1}^K\sum_{j\neq k}\left(2a_k\frac{n_j}{n}-\frac{\nk n_j}{n^2}\sum_{i=1}^{K}a_{i}\right)\Big|\model_2\right]\tau^2.\\
 \end{align*}
 
 Now consider the true variance, which we are trying to estimate.
 
 \begin{align*}
 \text{var}(\taufs|\model_2)&= \text{var}\left(\sum_{k=1}^K\frac{\nk}{n}\taukfs\Big|\model_2\right)\\
 &=\EE\left[\left(\sum_{k=1}^K\frac{\nk}{n}\taukfs - \tau\right)^2\Big|\model_2\right]\\
 &=\EE\left[\left(\sum_{k=1}^K\frac{\nk}{n}\taukfs\right)^2\Big|\model_2\right] - \tau^2\\
  &=\EE\left[\sum_{k=1}^K\frac{\nk^2}{n^2}\taukfs^2\Big|\model_2\right]+\EE\left[\sum_{k=1}^K\sum_{j\neq k}\frac{\nk n_j}{n^2}\taukfs\tau_{j,\mathcal{S}}\Big|\model_2\right] - \tau^2\\
    &=\EE\left[\sum_{k=1}^K\frac{\nk^2}{n^2}\taukfs^2\Big|\model_2\right]-\EE\left[\sum_{k=1}^K\frac{\nk^2}{n^2}\Big|\model_2\right]\tau^2
 \end{align*}
 
We have the last equality because
\begin{align*}
\EE\left[\sum_{k=1}^K\sum_{j\neq k}\frac{\nk n_j}{n^2}\taukfs\tau_{j,\mathcal{S}}\Big|\model_2\right] &= \EE\left[\sum_{k=1}^K\sum_{j\neq k}\frac{\nk n_j}{n}\Big|\model_2\right]\tau^2\\
&=\EE\left[\sum_{k=1}^K\frac{\nk}{n}(1-\frac{\nk}{n})\Big|\model_2\right]\tau^2 \\
&= \EE\left[1-\sum_{k=1}^K\frac{\nk^2}{n^2}\Big|\model_2\right]\tau^2.
\end{align*}

Matching this up with the expectation of our estimator, we want
\begin{align*}
\frac{\nk^2}{n^2}=a_{k}-2a_k\frac{\nk}{n}+\frac{\nk^2}{n^2}\sum_{i=1}^{K}a_{i}
\end{align*}

 and
 \begin{align*}
 \sum_{k=1}^K\frac{\nk^2}{n^2} = \sum_{k=1}^K\sum_{j\neq k}\left(2a_k\frac{n_j}{n}-\frac{\nk n_j}{n^2}\sum_{i=1}^{K}a_{i}\right).
 \end{align*}
 
 The first equation we solved for before in Section~\ref{append:var_est_small_p_derv}.
So we will get
 \[a_k = \frac{\nk^2}{(n-2\nk)(n+\sum_{j=1}^K\frac{n_j^2}{n-2n_j})}.\]

Let's see if this weight works for the second term.

\begin{align*}
&\sum_{k=1}^K\sum_{j\neq k}\left(2a_k\frac{n_j}{n}-\frac{\nk n_j}{n^2}\sum_{i=1}^{K}a_{i}\right)\\
&=\sum_{k=1}^K2a_k(1-\frac{\nk}{n})-(1-\sum_{k=1}^K\frac{\nk^2}{n^2})\sum_{i=1}^{K}a_{i}\\
&=\sum_{k=1}^K \frac{2\nk^2(1-\frac{\nk}{n})}{(n-2\nk)(n+\sum_{j=1}^K\frac{n_j^2}{n-2n_j})}-(1-\sum_{k=1}^K\frac{\nk^2}{n^2})\sum_{i=1}^{K} \frac{n_i^2}{(n-2n_i)(n+\sum_{j=1}^K\frac{n_j^2}{n-2n_j})}\\
%Trial
&=\sum_{k=1}^K \frac{\nk^2(1-2\frac{\nk}{n})}{(n-2\nk)(n+\sum_{j=1}^K\frac{n_j^2}{n-2n_j})}+\sum_{k=1}^K\frac{\nk^2}{n^2}\sum_{i=1}^{K} \frac{n_i^2}{(n-2n_i)(n+\sum_{j=1}^K\frac{n_j^2}{n-2n_j})}\\
&=\sum_{k=1}^K \frac{\nk^2(n-2\nk)}{n(n-2\nk)(n+\sum_{j=1}^K\frac{n_j^2}{n-2n_j})}+\sum_{k=1}^K\frac{\nk^2}{n^2}\sum_{i=1}^{K} \frac{n_i^2}{(n-2n_i)(n+\sum_{j=1}^K\frac{n_j^2}{n-2n_j})}\\
&=\sum_{k=1}^K \frac{\nk^2}{n(n+\sum_{j=1}^K\frac{n_j^2}{n-2n_j})}+\sum_{k=1}^K\frac{\nk^2}{n^2}\sum_{i=1}^{K} \frac{n_i^2}{(n-2n_i)(n+\sum_{j=1}^K\frac{n_j^2}{n-2n_j})}\\
&=\sum_{k=1}^K\frac{\nk^2(n+\sum_{i=1}^K \frac{n_i^2}{(n-2n_i)})}{n^2(n+\sum_{j=1}^K\frac{n_j^2}{n-2n_j})}\\
&=\sum_{k=1}^K\frac{\nk^2}{n^2}
\end{align*}

Which is exactly what we wanted.
So this weight works.
And because it is the same as the weight for the finite sample (where we wanted to get the first term in the variance decomposition correct), this also takes care of the first term in the variance decomposition.

The proof of this result for an infinite number of infinite size strata is direct by replacing the conditioning on $\mathcal{S}$ by conditioning on $\bm{B}$ and using results from the stratified sampling framework.
 \end{proof}

 \section{Creation of $\widehat{\sigma}^2_{SRS}$, Equation~(\ref{eq:var_est_srs_bk})}\label{append_var_srs}

In this section we derive a variance estimator for the case of simple random sampling and flexible blocks.
This framework includes the case of units coming from a simple random sample, with the blocks made after the fact.

We begin with the basic variance decomposition to examine what we are trying to estimate.
Let $\sigmatc$ be the population variance of treatment effects.

\begin{align}
 \text{var}(\tauestbk|SRS) &=\EE\left[\text{var}(\tauestbk|\mathcal{S})|SRS\right]+\text{var}\left(\EE[\tauestbk|\mathcal{S}]|SRS\right)\nonumber\\
 &=\EE\left[\sum_{k=1}^K\frac{\nk^2}{n^2}\left(\frac{\Sck}{\nck}+\frac{\Stk}{\ntk}-\frac{\Stck}{\nk}\right)\Big|SRS\right]+\text{var}\left(\taufs|SRS\right)\nonumber\\
  &=\EE\left[\sum_{k=1}^K\frac{\nk^2}{n^2}\left(\frac{\Sck}{\nck}+\frac{\Stk}{\ntk}-\frac{\Stck}{\nk}\right)\Big|SRS\right]+\frac{\sigmatc}{n}\nonumber\\
  &=\EE\left[\sum_{k=1}^K\frac{\nk^2}{n^2}\left(\frac{\Sck}{\nck}+\frac{\Stk}{\ntk}-\frac{\Stck}{\nk}\right)\Big|SRS\right]+\EE\left[\frac{\Stc}{n}\Big|SRS\right] \nonumber\\
    &=\underbrace{\EE\left[\sum_{k=1}^K\frac{\nk^2}{n^2}\left(\frac{\Sck}{\nck}+\frac{\Stk}{\ntk}\right)\Big|SRS\right]}_{\textbf{A}}+\underbrace{\EE\left[\frac{\Stc}{n}-\sum_{k=1}^K\frac{\nk}{n}\frac{\Stck}{n}\Big|SRS\right]}_{\textbf{B}}\label{eq_srs_breakdown}
\end{align}

Let's examine $\Stc$ so we can simplify term \textbf{B}.
\begin{align*}
\Stc &= \sum_{k=1}^K\frac{\nk-1}{n-1}\Stck + \sum_{k=1}^K\frac{\nk}{n-1}\left(\taukfs-\taufs\right)^2
\end{align*}

So term \textbf{B} can simplify as follows:
\begin{align*}
\Stc-\sum_{k=1}^K\frac{\nk}{n}\Stck &=\sum_{k=1}^K\frac{\nk-1}{n-1}\Stck + \sum_{k=1}^K\frac{\nk}{n-1}\left(\taukfs-\taufs\right)^2 -\sum_{k=1}^K\frac{\nk}{n}\Stck\\
&=\sum_{k=1}^K\frac{\nk}{n-1}\left(\taukfs-\taufs\right)^2 -\sum_{k=1}^K\frac{n-\nk}{n(n-1)}\Stck.
\end{align*}

Now recall from Supplementary Materials~\ref{append:var_est_small_p_derv} that
  \begin{align*}
\EE\left[\sum_{k=1}^K a_k \left(\taukest - \tauestbk\right)^2 |\mathcal{S}\right] 
 &=  \sum_{k=1}^{K}\left(a_{k} - 2a_{k}\frac{\nk}{n} + \frac{n_{k}^{2}}{n^{2}}\sum_{j=1}^{K}a_{j}\right)\text{var}\left(\taukest|\mathcal{S}\right) + \sum_{k=1}^{K}a_{k}\left(\taukfs - \taufs\right)^{2}.
 \end{align*}
 
 Letting $a_k = \nk$, we have
   \begin{align*}
\EE&\left[\sum_{k=1}^K \nk \left(\taukest - \tauestbk\right)^2 \Big|\mathcal{S}\right]\\
 &=  \sum_{k=1}^{K}\left(\nk - \frac{\nk^2}{n}\right)\text{var}\left(\taukest|\mathcal{S}\right) + \sum_{k=1}^{K}\nk\left(\taukfs - \taufs\right)^{2}\\
  &=  \sum_{k=1}^{K}\frac{\nk(n-\nk)}{n}\text{var}\left(\taukest|\mathcal{S}\right) + \sum_{k=1}^{K}\nk\left(\taukfs - \taufs\right)^{2}\\
    &=  \sum_{k=1}^{K}\frac{\nk(n-\nk)}{n}\left(\frac{\Sck}{\nck}+\frac{\Stk}{\ntk}-\frac{\Stck}{\nk}\right) + \sum_{k=1}^{K}\nk\left(\taukfs - \taufs\right)^{2}\\
  &=  \sum_{k=1}^{K}\frac{\nk(n-\nk)}{n}\left(\frac{\Sck}{\nck}+\frac{\Stk}{\ntk}\right) + \sum_{k=1}^{K}\nk\left(\taukfs - \taufs\right)^{2} -\sum_{k=1}^{K}\frac{(n-\nk)}{n}\Stck\\
    &=  \sum_{k=1}^{K}\frac{\nk(n-\nk)}{n}\left(\frac{\Sck}{\nck}+\frac{\Stk}{\ntk}\right) + (n-1)\left(\Stc-\sum_{k=1}^K\frac{\nk}{n}\Stck\right).
 \end{align*}
 
 This means that
    \begin{align*}
\EE&\left[\sum_{k=1}^K \frac{\nk}{n(n-1)} \left(\taukest - \tauestbk\right)^2 \Big|\mathcal{S}\right]
 =\sum_{k=1}^{K}\frac{\nk(n-\nk)}{n^2(n-1)}\left(\frac{\Sck}{\nck}+\frac{\Stk}{\ntk}\right) + \frac{\Stc}{n}-\sum_{k=1}^K\frac{\nk}{n}\frac{\Stck}{n}.
 \end{align*}
 
 So we have a way to estimate term \textbf{B} of Equation~\ref{eq_srs_breakdown}, which means we just need to add in a correction to get term \textbf{A}.
 \begin{align*}
& \sum_{k=1}^{K}\frac{\nk^2}{n^2}\left(\frac{\Sck}{\nck}+\frac{\Stk}{\ntk}\right)-\sum_{k=1}^{K}\frac{\nk(n-\nk)}{n^2(n-1)}\left(\frac{\Sck}{\nck}+\frac{\Stk}{\ntk}\right)\\
 &= \sum_{k=1}^{K}\frac{\nk(\nk-1)}{n(n-1)}\left(\frac{\Sck}{\nck}+\frac{\Stk}{\ntk}\right)
 \end{align*}
 
 Putting it all together, we have
 \begin{align*}
\EE\left[ \sum_{k=1}^{K}\frac{\nk(\nk-1)}{n(n-1)}\left(\frac{\Sck}{\nck}+\frac{\Stk}{\ntk}\right) + \sum_{k=1}^K \frac{\nk}{n(n-1)} \left(\taukest - \tauestbk\right)^2\Big|SRS\right]= \text{var}(\tauestbk|SRS).
 \end{align*}
 
 So 
 \[\widehat{\sigma}^2_{SRS} = \sum_{k=1}^{K}\frac{\nk(\nk-1)}{n(n-1)}\left(\frac{\sckest}{\nck}+\frac{\stkest}{\ntk}\right) + \sum_{k=1}^K \frac{\nk}{n(n-1)} \left(\taukest - \tauestbk\right)^2\]
 is an unbiased variance estimator.

%\bibliographystyle{apalike}
%\bibliography{poststratref}{}

%\end{document}

%% file: blocking_variance.bbl
\begin{thebibliography}{}

\bibitem[Abadie and Imbens, 2008]{abadie2008}
Abadie, A. and Imbens, G.~W. (2008).
\newblock Estimation of the conditional variance in paired experiments.
\newblock {\em Annales d'\'{E}conomie et de Statistique}, 91/92:175--187.

\bibitem[Aronow et~al., 2014]{aronow2014sharp}
Aronow, P.~M., Green, D.~P., Lee, D.~K., et~al. (2014).
\newblock Sharp bounds on the variance in randomized experiments.
\newblock {\em The Annals of Statistics}, 42(3):850--871.

\bibitem[{Centers for Disease Control and Prevention (CDC). National Center for
  Health Statistics (NCHS)}, 2014]{nhanes}
{Centers for Disease Control and Prevention (CDC). National Center for Health
  Statistics (NCHS)} (2013-2014).
\newblock {National Health and Nutrition Examination Survey Data}.
\newblock {\em Hyattsville, MD: U.S. Department of Health and Human Services,
  CDC}.

\bibitem[Cochran, 1953]{cochran_1953}
Cochran, W.~G. (1953).
\newblock Matching in analytical studies.
\newblock {\em American Journal of Public Health and the Nations Health},
  43(6\_Pt\_1):684--691.

\bibitem[Cochran, 1977]{CochranWilliamG.WilliamGemmell1977St}
Cochran, W.~G. (1977).
\newblock {\em Sampling techniques}.
\newblock {Wiley Series in Probability and Mathematical Statistics-Applied}.
  John Wiley \& Sons, New York, 3d edition.

\bibitem[Cochran and Cox, 1950]{cochran1950}
Cochran, W.~G. and Cox, G.~M. (1950).
\newblock {\em Experimental Designs}.
\newblock John Wiley \& Sons, New York, NY.

\bibitem[Ding et~al., 2017]{ding2017bridging}
Ding, P., Li, X., and Miratrix, L.~W. (2017).
\newblock Bridging finite and super population causal inference.
\newblock {\em Journal of Causal Inference}, 5(2).

\bibitem[Fisher, 1926]{Fisher1992}
Fisher, R.~A. (1926).
\newblock The arrangement of field experiments.
\newblock {\em Journal of Ministry of Agriculture}, 33:503--513.

\bibitem[Fogarty, 2018]{fogarty_2017}
Fogarty, C.~B. (2018).
\newblock On mitigating the analytical limitations of finely stratified
  experiments.
\newblock {\em J. Roy. Statist. Soc. Ser. B}, 80(5):1035--1056.

\bibitem[Freedman, 008a]{freedman_2008a}
Freedman, D.~A. (2008a).
\newblock On regression adjustments to experimental data.
\newblock {\em Advances in Applied Mathematics}, 40(2):180--193.

\bibitem[Freedman, 008b]{freedman_2008b}
Freedman, D.~A. (2008b).
\newblock On regression adjustments in experiments with several treatments.
\newblock {\em Ann. Appl. Stat.}, 2(1):176--196.

\bibitem[Gerber and Green, 2012]{gerber_green_text}
Gerber, A.~S. and Green, D.~P. (2012).
\newblock {\em Field Experiments: Design, Analysis and Interpretation}.
\newblock Norton, New York.

\bibitem[Hansen, 2004]{hansen_full_match}
Hansen, B.~B. (2004).
\newblock Full matching in an observational study of coaching for the {SAT}.
\newblock {\em J. Amer. Statist. Assoc.}, 99(467):609--618.

\bibitem[Hansen and Klopfer, 2006]{optmatch}
Hansen, B.~B. and Klopfer, S.~O. (2006).
\newblock Optimal full matching and related designs via network flows.
\newblock {\em J. Comput. Graph. Statist}, 15(3):609--627.

\bibitem[Hinkley, 1977]{hinkley1977jackknifing}
Hinkley, D.~V. (1977).
\newblock Jackknifing in unbalanced situations.
\newblock {\em Technometrics}, 19(3):285--292.

\bibitem[Iacus et~al., 2012]{iacus2012causal}
Iacus, S.~M., King, G., and Porro, G. (2012).
\newblock Causal inference without balance checking: Coarsened exact matching.
\newblock {\em Political Analysis}, 20(1):1--24.

\bibitem[Iacus et~al., 2016]{cem}
Iacus, S.~M., King, G., and Porro, G. (2016).
\newblock {\em cem: Coarsened Exact Matching}.
\newblock R package version 1.1.17.

\bibitem[Imai, 2008]{Imai2008}
Imai, K. (2008).
\newblock Variance identification and efficiency analysis in randomized
  experiments under the matched-pair design.
\newblock {\em Stat. Med.}, 27(24):4857--4873.

\bibitem[Imai et~al., 2008]{imai_king_stuart_2008}
Imai, K., King, G., and Stuart, E.~A. (2008).
\newblock Misunderstandings between experimentalists and observationalists
  about causal inference.
\newblock {\em J. Roy. Statist. Soc. Ser. A}, 171(2):481--502.

\bibitem[Imbens, 2011]{imbens2011experimental}
Imbens, G.~W. (2011).
\newblock Experimental design for unit and cluster randomid trials.
\newblock {\em Conf. International Initiative for Impact Evaluation,
  Cuernavaca}.

\bibitem[Imbens and Rubin, 2015]{CausalInferenceText}
Imbens, G.~W. and Rubin, D.~B. (2015).
\newblock {\em Causal Inference for Statistics, Social, and Biomedical
  Sciences: An Introduction}.
\newblock Cambridge University Press, New York.

\bibitem[Kosmidis, 2017]{brglm}
Kosmidis, I. (2017).
\newblock {\em {brglm}: Bias Reduction in Binary-Response Generalized Linear
  Models}.
\newblock R package version 0.6.1.

\bibitem[LaLonde, 1986]{lalonde1986evaluating}
LaLonde, R.~J. (1986).
\newblock Evaluating the econometric evaluations of training programs with
  experimental data.
\newblock {\em The American Economic Review}, 76(4):604--620.

\bibitem[Lin, 2013]{winston_lin}
Lin, W. (2013).
\newblock Agnostic notes on regression adjustments to experimental data:
  Reexamining {F}reedman's critique.
\newblock {\em Ann. Appl. Stat.}, 7(1):295--318.

\bibitem[Lohr, 2009]{lohr2010sampling}
Lohr, S.~L. (2009).
\newblock {\em Sampling: {D}esign and {A}nalysis}.
\newblock Cengage Learning, Boston, 2nd edition.

\bibitem[Miratix and Pashley, 2020]{blkvar}
Miratix, L.~W. and Pashley, N.~E. (2020).
\newblock blkvar.
\newblock \url{https://rdrr.io/github/lmiratrix/blkvar/}.

\bibitem[Miratrix et~al., 2020]{miratrix_weiss}
Miratrix, L., Weiss, M., and Henderson, B. (2020).
\newblock An applied researcher's guide to estimating effects from multisite
  individually randomized trials: Estimands, estimators, and estimates.
\newblock Working paper.

\bibitem[Miratrix et~al., 2018]{miratrix_survey}
Miratrix, L.~W., Sekhon, J.~S., Theodoridis, A.~G., and Campos, L.~F. (2018).
\newblock Worth weighting? {How} to think about and use weights in survey
  experiments.
\newblock {\em Political Analysis}, 26(3):275--291.

\bibitem[Miratrix et~al., 2013]{Miratrix2013}
Miratrix, L.~W., Sekhon, J.~S., and Yu, B. (2013).
\newblock Adjusting treatment effect estimates by post-stratification in
  randomized experiments.
\newblock {\em J. Roy. Statist. Soc. Ser. B}, 75(2):369--396.

\bibitem[Mukerjee et~al., 2018]{mukerjee2018using}
Mukerjee, R., Dasgupta, T., and Rubin, D.~B. (2018).
\newblock Using standard tools from finite population sampling to improve
  causal inference for complex experiments.
\newblock {\em J. Amer. Statistic. Assoc.}, 113(522):868--881.

\bibitem[Pashley and Miratrix, 2020]{pashley2020bkcr}
Pashley, N.~E. and Miratrix, L.~W. (2020).
\newblock Block what you can, except when you shouldn{'}t.
\newblock Working paper.

\bibitem[{R Core Team}, 2016]{Rsoft}
{R Core Team} (2016).
\newblock {\em R: A Language and Environment for Statistical Computing}.
\newblock R Foundation for Statistical Computing, Vienna, Austria.

\bibitem[Raudenbush and Schwartz, 2020]{raudenbush2020randomized}
Raudenbush, S.~W. and Schwartz, D. (2020).
\newblock Randomized experiments in education, with implications for multilevel
  causal inference.
\newblock {\em Annual Review of Statistics and Its Application}, 7:177--208.

\bibitem[Rosenbaum, 1991]{rosenbaum1991characterization}
Rosenbaum, P.~R. (1991).
\newblock A characterization of optimal designs for observational studies.
\newblock {\em J. Roy. Statist. Soc. Ser. B}, 53(3):597--610.

\bibitem[Rosenbaum, 2010]{rosenbaum_2010}
Rosenbaum, P.~R. (2010).
\newblock {\em Design of Observational Studies}.
\newblock Springer Series in Statistics. Springer, New York.

\bibitem[Rosenbaum and Zhao, 2017]{crossscreening}
Rosenbaum, P.~R. and Zhao, Q. (2017).
\newblock {\em CrossScreening: Cross-Screening in Observational Studies that
  Test Many Hypotheses}.
\newblock R package version 0.1.1.

\bibitem[Rubin, 1974]{rubin_1974}
Rubin, D.~B. (1974).
\newblock Estimating causal effects of treatments in randomized and
  nonrandomized studies.
\newblock {\em Journal of Educational Psychology}, 66(5):688--701.

\bibitem[Rubin, 1980]{rubin_1980}
Rubin, D.~B. (1980).
\newblock Randomization analysis of experimental data: {T}he {F}isher
  randomization test comment.
\newblock {\em J. Amer. Statist. Assoc.}, 75(371):591--593.

\bibitem[S{\"a}rndal et~al., 2003]{sarndal2003model}
S{\"a}rndal, C.-E., Swensson, B., and Wretman, J. (2003).
\newblock {\em Model assisted survey sampling}.
\newblock Springer, New York.

\bibitem[S{\"a}vje, 2015]{savje2015performance}
S{\"a}vje, F. (2015).
\newblock The performance and efficiency of threshold blocking.
\newblock {\em arXiv preprint arXiv:1506.02824}.

\bibitem[Schochet, 2016]{schochet_2016}
Schochet, P.~Z. (2016).
\newblock Statistical theory for the {RCT-YES} software: Design-based causal
  inference for {RCT}s, {S}econd {E}dition.
\newblock Technical Report (NCEE 2015-4011), Washington, DC: U.S. Department of
  Education, Institute of Education Sciences, National Center for Education
  Evaluation and Regional Assistance, Analytic Technical Assistance and
  Development.

\bibitem[Scosyrev, 2014]{scosyrev_2014}
Scosyrev, E. (2014).
\newblock Causal inference in block-randomized experiments: Analysis based on
  {Neyman's} stochastic causal model.
\newblock Unpublished.

\bibitem[Splawa-Neyman et~al., 1990]{neyman_1923}
Splawa-Neyman, J., Dabrowska, D.~M., and Speed, T. (1923/1990).
\newblock On the application of probability theory to agricultural experiments.
  {E}ssay on principles. {S}ection 9.
\newblock {\em Statist. Sci.}, 5(4):465--472.

\bibitem[StataCorp, 2017]{statacorp2017stata}
StataCorp (2017).
\newblock Stata {S}tatistical {S}oftware: {R}elease 15.
\newblock {\em StataCorp. College Station, TX: StataCorp}.

\bibitem[Wu and Hamada, 2000]{WuHamada}
Wu, C. F.~J. and Hamada, M.~S. (2000).
\newblock {\em Experiments : Planning, Analysis, and Parameter Design
  Optimization}.
\newblock John Wiley \& Sons, New York, NY.

\bibitem[Zhao et~al., 2018]{zhao2017cross}
Zhao, Q., Small, D.~S., and Rosenbaum, P.~R. (2018).
\newblock Cross-screening in observational studies that test many hypotheses.
\newblock {\em J. Amer. Statist. Assoc.}, 113(523):1070--1084.

\end{thebibliography}
